\documentclass[11pt,a4paper]{article}
\usepackage{amssymb}
\usepackage{amsmath}
\usepackage{amsfonts}
\usepackage{bm}
\usepackage{bbm}
\usepackage{amsthm}
\usepackage{mathrsfs}
\usepackage{hyperref}
\usepackage{color}
\usepackage[margin=2.41cm]{geometry}
\usepackage[all,cmtip]{xy}
\usepackage[utf8]{inputenc}
\usepackage{graphicx}
\usepackage{varwidth}
\usepackage{comment}
\usepackage{enumitem}

\usepackage{mathtools}

\usepackage{upgreek}
\usepackage{rotating}

\usepackage{tikz}
\usetikzlibrary{shapes.geometric}
\usetikzlibrary{patterns.meta}
\usetikzlibrary{cd}
\usetikzlibrary{arrows}

\definecolor{darkred}{rgb}{0.8,0.1,0.1}
\hypersetup{
	colorlinks=true,         % false: boxed links; true: colored links,false is default
	urlcolor=darkred,
	linkcolor=darkred,
	citecolor=blue,
}

\theoremstyle{plain}
\newtheorem{theo}{Theorem}[section]

\newtheorem{propo}[theo]{Proposition}

\theoremstyle{definition}
\newtheorem{defi}[theo]{Definition}
\newtheorem{assu}[theo]{Assumption}

\newtheorem{exx}[theo]{Example}
\newenvironment{ex}
{\pushQED{\qed}\exx}
{\popQED\endexx}

\newtheorem{remm}[theo]{Remark}
\newenvironment{rem}
{\pushQED{\qed}\remm}
{\popQED\endremm}

\numberwithin{equation}{section}

\def\nn{\nonumber}

\def\bbR{\mathbb{R}}
\def\bbC{\mathbb{C}}
\def\bbN{\mathbb{N}}
\def\bbZ{\mathbb{Z}}

\def\id{\mathrm{id}}

\def\deg{\mathrm{deg}}
\def\dd{\mathrm{d}}

\def\cc{\mathrm{c}}

\def\dim{\mathrm{dim}}
\def\1{I}

\def\ad{\mathrm{ad}}

\def\FFF{\mathfrak{F}}

\def\A{\mathcal{A}}

\def\E{\mathcal{E}}
\def\F{\mathcal{F}}

\def\L{\mathcal{L}}
\def\M{\mathcal{M}}
\def\O{\mathcal{O}}

\newcommand\und[1]{\underline{#1}}

\newcommand{\pair}[2]{\langle\!\langle #1 , #2 \rangle\!\rangle}

\def\g{\mathfrak{g}}
\def\h{\mathfrak{h}}

\def\CP{\mathbb{C}P^1}
\def\delbar{\overline{\partial}}

\def\sk{\vspace{2mm}}

\makeatletter
\let\@fnsymbol\@alph
\makeatother

\title{%
On the structure of higher-dimensional integrable field theories
}

\author{%
Marco Benini$^{1,2,a}$, Ryan A.~Cullinan$^{3,b}$, Alexander Schenkel$^{4,c}$\ and\ Beno\^{\i}t Vicedo$^{3,d}$\vspace{4mm}\\
{\small ${}^1$ Dipartimento di Matematica, Dipartimento di Eccellenza 2023-27, Universit\`a di Genova,}\\
{\small Via Dodecaneso 35, 16146 Genova, Italy.}\vspace{2mm}\\
{\small ${}^2$ INFN, Sezione di Genova,}\\
{\small Via Dodecaneso 33, 16146 Genova, Italy.}\vspace{2mm}\\
{\small ${}^3$ Department of Mathematics, University of York,}\\
{\small Heslington, York YO10 5GH, United Kingdom.}\vspace{2mm}\\
{\small ${}^4$~Dipartimento di Matematica, Universit{\`a} di Trento and INFN-TIFPA,}\\
{\small Via Sommarive 14, 38123 Povo (Trento), Italy.}\vspace{4mm}\\
{\small \begin{tabular}{ll}
Email: & ${}^a$~\href{mailto:marco.benini@unige.it}{\texttt{marco.benini@unige.it}}\\
& ${}^b$~\href{mailto:ryan.cullinan@york.ac.uk}{\texttt{ryan.cullinan@york.ac.uk}}\\
& ${}^c$~\href{mailto:alexander.schenkel@unitn.it}{\texttt{alexander.schenkel@unitn.it}}\\
& ${}^d$~\href{mailto:benoit.vicedo@gmail.com}{\texttt{benoit.vicedo@gmail.com}}
\vspace{2mm}
\end{tabular}
}
}

\date{April 2026}

\begin{document}

\maketitle

\begin{abstract}
\noindent We propose a general framework for integrable field theories in arbitrary spacetime dimension $d+1$ which is based on $d$-term $L_\infty$-algebras. Specifically, we introduce cyclic $L_\infty$-algebras describing topological-holomorphic higher Chern-Simons theories on $M \times \mathbb{C}P^1$ with suitable singularity structures and boundary conditions, controlled by a meromorphic $1$-form on $\mathbb{C}P^1$. Using homological perturbation theory and homotopy transfer, we construct weakly equivalent models describing $(d+1)$-dimensional field theories on $M$. Their integrability is witnessed by a natural map to an $L_\infty$-algebra describing higher Lax connections, yielding conserved charges associated with higher-dimensional cycles in $M$. The resulting theories admit natural action functionals and recover the Costello-Yamazaki construction in $2$ dimensions.
\end{abstract}
\vspace{-1mm}

\paragraph*{Keywords:} topological-holomorphic higher Chern-Simons theories, 
higher-dimensional integrable field theories, $L_\infty$-algebras
\vspace{-2mm}

\paragraph*{MSC 2020:} 70Sxx, 81Txx, 55Uxx
\vspace{-2mm}

\renewcommand{\baselinestretch}{0.8}\normalsize
\tableofcontents
\renewcommand{\baselinestretch}{1.0}\normalsize

%%%%%%%%%%%%%%%%%%%%%%%%%%%%%%%%%%%%%%%%%%%%%%%%
%%%%%%%%%%%%%%%%%%%%%%%%%%%%%%%%%%%%%%%%%%%%%%%%

\section{\label{sec:intro}Introduction and summary}
One of the hallmarks of integrable field theories is that their non-linear 
equations of motion admit a reformulation as the compatibility condition for an 
over-determined auxiliary linear system. This observation underpins the Lax 
formalism \cite{Lax}, which guarantees the existence of infinite towers of conserved 
charges and allows the non-linear dynamics to be linearized via the method of 
inverse scattering \cite{Gardner:1967wc, Ablowitz:1974ry}. Typically, the Lax operators 
of the auxiliary linear system depend meromorphically on a complex variable, 
known as the spectral parameter, and take values in a finite-dimensional Lie algebra, 
hinting at a deeper geometric origin of these objects. Indeed, this Lie-algebraic 
structure naturally suggests an interpretation of the Lax formalism in terms of 
gauge theory, where the Lax operators arise as covariant derivatives associated with a connection.
\sk

A striking realization of this idea has emerged in recent years through 
the pioneering work of Costello and Yamazaki \cite{CY3}. They showed that 
$2$-dimensional integrable field theories on a spacetime manifold $\Sigma$ 
can be constructed from a $4$-dimensional topological-holomorphic Chern-Simons 
theory on the product manifold $\Sigma \times \CP$ after imposing suitable 
singularity structures and boundary conditions on the $4$-dimensional gauge 
field along surface defects located at specific marked points in $\CP$, see 
also \cite{DLMV,BSV,Lacroix,BSV2} for more details.
From this perspective, the auxiliary linear system is reinterpreted as the 
equation for a flat section over the topological directions $\Sigma$, while 
the spectral parameter becomes the complex coordinate on $\CP$. In this way, 
the Lax formalism emerges directly from a higher-dimensional gauge theory, 
with both the Lax connection and spectral parameter acquiring a unified geometric origin.
\sk

Remarkably, this geometric paradigm extends beyond $2$-dimensional 
integrable field theories. Indeed, the $4$-dimensional anti-self-dual Yang-Mills 
(ASDYM) equation is a celebrated master integrable system giving rise to many 
lower-dimensional integrable models by symmetry reduction. Bittleston and Skinner 
showed in \cite{Bittleston:2020hfv} that it can be obtained from $6$-dimensional 
holomorphic Chern-Simons theory on projective twistor space $\mathbb{P}\mathbb{T}$ 
with suitable boundary conditions imposed on the $6$-dimensional holomorphic 
gauge field along codimension $2$ defects. In this setting, the spectral $\CP$ 
is elegantly realized as the twistor lines of the underlying twistor geometry.
\sk

The $4$-dimensional topological-holomorphic Chern-Simons theory 
of Costello and Yamazaki arises as a symmetry reduction of this 
$6$-dimensional holomorphic Chern-Simons theory \cite{Bittleston:2020hfv}. 
Thus, $2$-dimensional integrable field theories descend from the latter 
in two separate ways: via $4$-dimensional topological-holomorphic 
Chern-Simons theory and as reductions of $4$-dimensional ASDYM equations and 
generalizations thereof. This `diamond' of connections to $6$-dimensional 
holomorphic Chern-Simons theory has been explored for various $2$-dimensional 
integrable field theories in \cite{Cole:2023umd, Cole:2024sje, Chatzis:2025lly, Ashwinkumar:2026zrr}.
There is also an intermediate symmetry reduction to a $5$-dimensional partially 
holomorphic Chern-Simons theory, with one topological and two holomorphic 
directions \cite{Bittleston:2020hfv}. Building upon this setup, the famous 
$3$-dimensional Kadomtsev-Petviashvili (KP) equation was recently obtained in 
\cite{Bittleston:2025gxr} from a non-commutative deformation of a 
$5$-dimensional partially holomorphic Poisson-Chern-Simons theory.
\sk

This series of works on $6$-dimensional holomorphic Chern-Simons theories, 
and their various reductions and deformations, highlights a fundamental dichotomy 
in the geometric origin of Lax operators. If a higher-dimensional 
topological-holomorphic Chern-Simons theory is to encode Lax integrable 
field theories then the underlying manifold on which it is formulated 
should contain at least one holomorphic direction corresponding to the spectral $\CP$. 
However, as articulated in the interesting recent work \cite{Cole:2025zmq}, 
the character of the resulting Lax integrability is then determined 
by the nature of the remaining directions.
\sk

If all other directions are topological, then the setup yields what one can refer to as a 
`topologically integrable' field theory, characterized by infinite towers of 
conserved charges arising from holonomies of flat connections along non-trivial cycles
which are invariant under continuous deformations. The prototypical
examples are $2$-dimensional integrable field theories arising from $4$-dimensional
topological-holomorphic Chern-Simons theory. Another edge case 
is furnished by $1$-dimensional integrable systems, such as the 
finite-dimensional Gaudin model or more generally the Hitchin system, 
which arise from $3$-dimensional mixed BF theory 
\cite{Vicedo:2022mrm, Caudrelier:2025xtx}.
\sk

If, on the other hand, some of the remaining directions are taken 
to be holomorphic, then one obtains what can be referred to as `holomorphically 
integrable' field theories, following \cite{Cole:2025zmq}. In this case,
the physical spacetime itself carries an intrinsic complex structure, 
and the associated auxiliary linear system is formulated in terms of 
anti-holomorphic derivatives with respect to distinguished spacetime coordinates 
which depend on the spectral parameter. This class includes the $4$-dimensional 
ASDYM equations and their generalizations, the $3$-dimensional 
integrable field theories derived from $5$-dimensional partially holomorphic 
Chern-Simons theory \cite{Bittleston:2020hfv}, and the novel type 
of $2$-dimensional integrable field theories recently obtained 
from $4$-dimensional holomorphic BF theory in \cite{Cole:2025zmq}.
Much like ASDYM itself, these `holomorphically integrable' field theories 
do not admit towers of topological conserved charges supported on non-trivial cycles. 
Instead, their characteristic feature is that they serve as higher-dimensional 
master theories admitting reductions to lower-dimensional 
topologically integrable field theories. The underlying reason is that 
their Lax description is still governed by an ordinary $1$-form connection.
\sk

Indeed, $6$-dimensional holomorphic Chern-Simons theory on 
$\mathbb{P}\mathbb{T}$, while providing a unifying geometric framework for 
integrability in dimensions up to $4$, remains an ordinary gauge theory. 
As such, the integrable field theories it produces
are governed by Lie algebra-valued $1$-form Lax connections, 
whose holonomies are intrinsically supported on $1$-dimensional cycles. 
This is sufficient in $2$ dimensions, where integrals of motion are naturally 
realized as line operators, and it underlies the role of ASDYM as a master theory 
for lower-dimensional integrable systems. However, in higher-dimensional topologically
integrable field theories, the topological conserved charges should be supported on higher-dimensional 
cycles, in particular of codimension $1$. From this perspective, $1$-form connections 
are fundamentally inadequate as their associated holonomies cannot capture the 
expected structure of higher-dimensional topological conserved charges. This mismatch signals 
a genuine geometric obstruction to extending the Lax formalism for topologically integrable
field theories beyond $2$ dimensions within the framework of ordinary gauge theory.
\sk

This obstruction points to the need for a genuinely higher-gauge-theoretic generalization. 
While the idea of using higher connections to describe higher-dimensional integrability
has already emerged in \cite{Alvarez:1997ma}, attempts towards a concrete realization 
only appeared very recently in the work \cite{SV} by two of the authors and also in \cite{CL}.
These works introduced a $5$-dimensional topological-holomorphic $2$-Chern-Simons theory 
on the product $M \times \CP$ of a $3$-dimensional manifold $M$ and the spectral $\CP$.
The absence of holomorphic spacetime directions places this model firmly in the class expected 
to yield $3$-dimensional topologically integrable field theories, making it a 
natural higher-dimensional analogue of the original framework of Costello and Yamazaki \cite{CY3}.
Unfortunately, these initial formulations encountered a structural limitation. 
While a consistent higher-gauge-theoretic description should incorporate both $1$- and $2$-gauge transformations, 
doing so in the usual higher-gauge-theoretic setting forces the fake-flatness condition, 
which severely constrains the curvature of the $1$-form connection and thereby degenerates
the desired action functionals. To avoid these issues and retain a
traditional action formulation, these early attempts were unable to take into account the 
full $2$-gauge symmetry, which made it difficult to properly interpret and analyze 
the resulting $3$-dimensional integrable field theories.
\sk

A natural resolution of this tension between incorporating higher gauge symmetry 
and avoiding the fake-flatness constraint is to pass from strict to homotopy-coherent 
algebraic structures. In recent work \cite{BSV2} by three of the authors, this 
homotopy-theoretic framework was developed in the context of $4$-dimensional 
topological-holomorphic Chern-Simons theory with general singularity structures 
and local boundary conditions, by formulating the theory in the powerful
language of $L_\infty$-algebras.
Within this framework, the construction of $2$-dimensional topologically integrable field theories 
can be understood as a homotopy transfer to the cohomology along the spectral $\CP$, 
thereby providing a rigorous homotopy-algebraic formulation of the equivalence 
between the $4$-dimensional gauge theory and its associated $2$-dimensional integrable field theory.
\sk

The purpose of the present paper is to extend this perspective 
to higher dimensions. Specifically, we seek a formulation that both accommodates 
higher gauge structures beyond ordinary Lie algebras and describes topologically integrable field theories 
in spacetime dimensions beyond $2$. Our proposal is that these generalizations are both 
controlled by finite-dimensional $d$-term $L_\infty$-algebras $(\mathfrak g,\ell^{\mathfrak g})$.
In particular, these structures will simultaneously encode the `higher structure Lie algebras' for 
$(d+3)$-dimensional topological-holomorphic Chern-Simons theories, capture the algebraic properties of 
higher Lax connections, and organize the field content, (higher) gauge symmetries
and interaction terms of $(d+1)$-dimensional topologically integrable field theories.
\sk

We now outline the structure of the paper and describe our main results. 
We begin in Section \ref{sec:prelim} with a brief review of $L_\infty$-algebras, 
their tensor products with commutative dg-algebras, and cyclic structures. 
In particular, we recall their standard interpretation in the context of field theory: 
The fields are given by degree $1$ elements, whose dynamics are governed 
by the Maurer-Cartan equation, while degree $0$ elements describe infinitesimal 
gauge transformations and elements of negative degree encode higher gauge symmetries. 
The elements of degree $\geq 2$ capture the corresponding anti-fields from the BV formalism.
In this way, an $L_\infty$-algebra packages the field content, 
equations of motion, and (higher) gauge structure of a field theory 
in a homotopy-coherent manner. Importantly, a cyclic structure of degree $-3$ 
on the $L_\infty$-algebra is required to define an action functional.
\sk

In Section \ref{sec:higherCS}, we introduce for any choice 
of $d$-term $L_\infty$-algebra $(\mathfrak g,\ell^{\mathfrak g})$
and $(d+1)$-dimensional manifold $M$ a corresponding topological-holomorphic 
$(\mathfrak g,\ell^{\mathfrak g})$-Chern-Simons theory without singularities 
and boundary conditions on $X = M \times \CP$. This is described by
the $L_\infty$-algebra $(\E(X),\ell) = (\mathfrak g,\ell^{\mathfrak g})\otimes \Omega^{\bullet,(0,\bullet)}(X)$ 
obtained by tensoring $(\mathfrak g,\ell^{\mathfrak g})$ 
with the de Rham-Dolbeault commutative dg-algebra on $X$. This $L_\infty$-algebra does not carry a cyclic structure
and therefore does not yet define a theory with an action functional. In order to construct a 
cyclic structure, and hence define the associated action functional, 
we will suitably modify the $L_\infty$-algebra $(\E(X), \ell)$ by incorporating 
singularities and imposing boundary conditions. These are both controlled by a 
choice of meromorphic $1$-form $\omega$ on $\CP$.
\sk

The incorporation of singularities is formalized in Section \ref{sec:higherLax} 
via the notion of a singularity structure introduced in Definition \ref{def:singularity}. 
This consists of a family $\mathscr{D}$ of non-negative divisors on $\CP$ satisfying 
compatibility conditions determined by the $L_\infty$-brackets and a chosen 
cyclic structure on $(\mathfrak g,\ell^{\mathfrak g})$. To such a singularity structure 
$\mathscr{D}$ we associate an $L_\infty$-algebra $(\E_{\mathscr{D}}(X), \ell)$
in which the various field components are allowed prescribed singularities, 
supported within the divisor $(\omega)_0$ of zeros of the meromorphic 
$1$-form $\omega$. We show in Proposition \ref{prop:higherLax} that this $L_\infty$-algebra 
is weakly equivalent to the $L_\infty$-algebra $(\L(M),\ell^\prime)$ of higher Lax 
connections on $M$, whose Maurer-Cartan elements are flat 
$(\mathfrak g,\ell^{\mathfrak g})$-connections on $M$ depending meromorphically 
on the spectral parameter, with poles controlled by the singularity structure $\mathscr{D}$.
\sk

The imposition of boundary conditions is subsequently formalized 
in Section \ref{sec:higherIFT} via the notion of a local boundary condition 
introduced in Definition \ref{def:bdycondition}. This consists of a 
family $\mathscr{B}$ of non-positive divisors on $\CP$, supported within the 
divisor $(\omega)_\infty$ of poles of the meromorphic $1$-form $\omega$, 
satisfying compatibility conditions parallel to those for singularity structures. 
Additionally, given a singularity structure $\mathscr{D}$, the local boundary condition 
$\mathscr{B}$ is required to satisfy a degree constraint which links it 
to the chosen singularity structure $\mathscr{D}$ in a way that will ensure 
(a higher generalization of) the usual relation between the equations of motion
of an integrable field theory and its Lax formulation.
To any such pair $(\mathscr{D},\mathscr{B})$ we associate an 
$L_\infty$-algebra $(\E_{\mathscr{D}+\mathscr{B}}(X), \ell)$ describing 
a topological-holomorphic $(\mathfrak g,\ell^{\mathfrak g})$-Chern-Simons theory 
on $X$ with singularities and boundary conditions. This $L_\infty$-algebra 
admits a natural cyclic structure of degree $-3$ on its compactly supported 
sections, providing an action functional for the resulting theory. For $d=1$, 
we recover the usual topological-holomorphic Chern-Simons theory of Costello and Yamazaki, 
while for $d \geq 2$ we obtain genuinely higher-gauge-theoretic and higher-dimensional analogues.
\sk

In Proposition \ref{prop:IFT} we use homological perturbation theory, 
together with homotopy transfer to the cohomology along the spectral 
$\CP$, to construct a weakly equivalent model $(\mathcal{F}(M), \ell^\prime)$ 
defined directly on $M$ for the $L_\infty$-algebra $(\E_{\mathscr{D}+\mathscr{B}}(X), \ell)$. 
The key structural result of Section \ref{sec:higherIFT} is that this model 
describes a $(d+1)$-dimensional topologically integrable field theory on $M$. Specifically, 
in Theorem \ref{theo:mapLax} we construct a weak $L_\infty$-morphism from 
$(\mathcal{F}(M), \ell^\prime)$ to the $L_\infty$-algebra $(\mathcal{L}(M), \ell^\prime)$ 
of higher Lax connections on $M$. Consequently, every Maurer-Cartan element of the 
transferred theory determines a higher Lax connection. This generalizes the usual 
relation between the equations of motion of a topologically integrable field theory and its 
Lax formulation from $2$ dimensions to arbitrary spacetime dimension $d+1$. 
Moreover, the higher parallel transport associated with such flat higher Lax 
connections provides a natural source of topological conserved charges, now attached not 
only to loops but also to higher-dimensional cycles in $M$. Finally, for a suitable 
choice of deformation retract data, the above cyclic structure transfers to the 
compactly supported sections of the $L_\infty$-algebra on $M$, so that the 
resulting $(d+1)$-dimensional topologically integrable field theory also admits a natural action functional.
\sk

To illustrate the framework, we provide in Section \ref{sec:examples} explicit 
families of singularity structures and local boundary conditions for arbitrary $d$. 
We then specialize to $d=2$, corresponding to $5$-dimensional topological-holomorphic 
$2$-Chern-Simons theory and its associated $3$-dimensional topologically integrable field theories. 
In this case, choosing as input datum the shifted tangent $L_\infty$-algebra $T[1]\mathfrak h$ 
of an ordinary Lie algebra $\mathfrak h$, we analyze the resulting
linearized dynamics of the corresponding $3$-dimensional topologically integrable field
theories $(\mathcal{F}(M), \ell^\prime)$. This produces examples of over-determined, 
under-determined and well-posed systems of gauge-theoretic partial differential equations, 
thereby giving a first indication of the range of behaviors encompassed by the general formalism. 
Remarkably, the latter example of a well-posed system of gauge-theoretic 
partial differential equations is controlled by a wave equation 
on the $3$-dimensional Minkowski spacetime.

%%%%%%%%%%%%%%%%%%%%%%%%%%%%%%%%%%%%%%%%%%%%%%%%
%%%%%%%%%%%%%%%%%%%%%%%%%%%%%%%%%%%%%%%%%%%%%%%%

\section{\label{sec:prelim}Preliminaries}
In this section we recall some basic definitions 
and constructions concerning $L_\infty$-algebras.
We refer the reader to \cite{KraftSchnitzer} for a 
comprehensive review of this subject and also to 
\cite{JRSW} for a more physicist-friendly introduction.
\begin{defi}\label{def:Linftyalgebra}
An \textit{$L_\infty$-algebra} is a pair $(L,\ell)$ consisting
of a $\bbZ$-graded vector space $L = (L^i)_{i\in\bbZ}$ over $\bbC$ 
and a family $\ell = (\ell_n)_{n\geq 1} = (\ell_1,\ell_2,\ell_3,\dots)$ 
of graded antisymmetric linear maps $\ell_n : L^{\otimes n}\to L$ of degree
$\vert \ell_n\vert = 2-n$ which satisfy the homotopy Jacobi identities
\begin{flalign}\label{eqn:Linftyalgebra}
\sum_{k+l-1=n}~ \sum_{\sigma\in\mathrm{Sh}(l,k-1)} (-1)^{\vert\sigma\vert} ~(-1)^{k-1}~ 
\ell_k \circ \left(\ell_l \otimes \id^{\otimes (k-1)}\right)\circ \gamma_{\sigma}~=~0\quad,
\end{flalign}
for all $n\geq 1$. Here we denote by $\mathrm{Sh}(l,k-1)\subseteq \Sigma_{n}$
the set of all $(l,k-1)$-shuffle permutations, by $(-1)^{\vert\sigma\vert}$
the parity of the permutation $\sigma\in \mathrm{Sh}(l,k-1)$, and by
$\gamma_\sigma : L^{\otimes n}\to L^{\otimes n}$ its action on tensor powers via
the usual symmetric braiding on graded vector spaces involving Koszul signs.
\end{defi}

Let us also recall the following standard construction
of tensor products of $L_\infty$-algebras and commutative dg-algebras.
\begin{propo}\label{prop:tensorproduct}
Let $(L,\ell)$ be an $L_\infty$-algebra and $A$ a commutative dg-algebra,
with differential denoted by $\dd:A\to A$ and $n$-ary multiplication by 
$\mu_n : A^{\otimes n} \to A\,,~a_1\otimes\cdots\otimes a_n\mapsto
a_1\cdots a_n$, for all $n\geq 2$. Then the tensor product $L\otimes A$ of $\bbZ$-graded 
vector spaces carries the structure of an $L_\infty$-algebra with brackets
given by
\begin{subequations}
\begin{flalign}
\xymatrix{
\ell_1^\otimes\,:\,L\otimes A
\ar[rr]^-{\ell_1\otimes \id + \id\otimes \dd}~&&~ L\otimes A
}
\end{flalign}
and
\begin{flalign}
\xymatrix{
\ell_n^\otimes\,:\, \big(L\otimes A\big)^{\otimes n} \ar[r]^-{\cong} ~&~
L^{\otimes n}\otimes A^{\otimes n}\ar[rr]^-{\ell_n\otimes\mu_n }~&&~L\otimes A
}\quad,
\end{flalign}
\end{subequations}
for all $n\geq 2$, where the unlabeled isomorphism is given by permuting
via the symmetric braiding on graded vector spaces
the tensor factors to the displayed form, without altering the order among the 
individual $L$'s.
\end{propo}

\begin{rem}\label{rem:tensorproduct}
With a slight abuse of notation, we will sometimes denote the $L_\infty$-algebra structure
$\ell^\otimes= (\ell^\otimes_n)_{n\geq 1}$ on $L\otimes A$ from Proposition \ref{prop:tensorproduct} simply by
\begin{flalign}
\ell^\otimes_1\,=\, \ell_1 + \dd\quad,\qquad
\ell^\otimes_n \,=\, \ell_n\quad,
\end{flalign}
for all $n\geq 2$.
\end{rem}

\begin{rem}\label{rem:Linftyalgebra}
We shall use $L_\infty$-algebras for the following
two purposes: 1.)~As an algebraic model for certain geometric objects
called \textit{formal moduli problems} \cite{Lurie,Pridham}, and 
2.)~as a model for the higher Lie algebra underlying the higher structure Lie group
of a higher gauge theory. Let us provide some comments and clarifications about these two contexts.
\sk

In the context of formal moduli problems, a key role
is played by the \textit{Maurer-Cartan elements} in the $L_\infty$-algebra 
$(L,\ell)$, which are elements $\alpha\in L^1$ of degree $1$ 
satisfying the Maurer-Cartan equation
\begin{flalign}\label{eqn:MCequation}
\sum_{n\geq 1}\tfrac{1}{n!}\,\ell_n(\alpha^{\otimes n})\,=\,0\quad.
\end{flalign}
One interprets such Maurer-Cartan elements as the points of 
the formal moduli problem associated with $(L,\ell)$.\footnote{For a rigorous 
construction of the formal moduli problem, one has to parametrize over 
Artinian dg-algebras, resulting in a functor of points $\mathbf{dgArt}^{\leq 0}\to \mathbf{sSet}$
which assigns to $A\in\mathbf{dgArt}^{\leq 0}$ the space of Maurer-Cartan elements
of the tensor product $L_\infty$-algebra $(L,\ell)\otimes \mathsf{m}(A)$, where $\mathsf{m}(A)$ denotes the maximal
ideal of $A$. While this parametrization over $\mathbf{dgArt}^{\leq 0}$ is crucial from a technical point of view, 
e.g.\ it truncates \eqref{eqn:MCequation} and \eqref{eqn:MCgauge} to finite sums, the underlying ideas of Maurer-Cartan theory
are more easily conveyed by ignoring this parametrization.} The degree $\leq 0$ components of the
$L_\infty$-algebra $(L,\ell)$ encode a tower of (higher) gauge transformations 
acting on such Maurer-Cartan elements. Concretely, an element $c_{(0)}\in L^0$ 
of degree $0$ acts as a gauge transformation on Maurer-Cartan elements $\alpha \in L^1$ 
according to
\begin{subequations}\label{eqn:MCgauge}
\begin{flalign}
\delta_{c_{(0)}}\alpha \,= \, \sum_{n\geq 0}\tfrac{1}{n!}\,\ell_{n+1}\big(\alpha^{\otimes n}\otimes c_{(0)}\big)\quad,
\end{flalign}
while an element $c_{(-1)}\in L^{-1}$ of degree $-1$ acts as a $2$-gauge transformation on $c_{(0)}\in L^0$ via
\begin{flalign}
\delta_{c_{(-1)}}c_{(0)} \,=\,\sum_{n\geq 0}\tfrac{1}{n!}\,\ell_{n+1}\big(\alpha^{\otimes n}\otimes c_{(-1)}\big) \quad.
\end{flalign}
In general, an element $c_{(1-k)}\in L^{1-k}$ of degree $1-k$, with $k\geq 2$, 
acts as a $k$-gauge transformation on the $(k{-}1)$-gauge transformation
$c_{(2-k)}\in L^{2-k}$ according to
\begin{flalign}
\delta_{c_{(1-k)}}c_{(2-k)}\,=\, \sum_{n\geq 0}\tfrac{1}{n!}\,\ell_{n+1}\big(\alpha^{\otimes n}\otimes c_{(1-k)}\big) \quad.
\end{flalign}
\end{subequations}
We refer the reader to \cite[Section 4]{JRSW} for a hands-on derivation of these formulas.
\sk

In the context of higher structure Lie algebras for higher gauge theories,
we denote the relevant $L_\infty$-algebras by symbols like $(\g,\ell^\g)$
and assume that they are of finite total dimension $\dim(\g) := \sum_{i\in\bbZ} \dim(\g^i) <\infty$.
For a higher gauge theory modeled on a Lie $N$-group, the corresponding
higher Lie algebra will be an $L_\infty$-algebra concentrated in degrees $\{1-N,\dots,-1,0\}$, i.e.\ 
its underlying cochain complex takes the form
\begin{flalign}\label{eqn:LieNalgebra}
\g\,=\, \Big(\xymatrix{
\g^{1-N} \ar[r]^-{\ell_1^\g}~&~\cdots \ar[r]^-{\ell_1^\g}~&~ \g^{-1} \ar[r]^-{\ell_1^\g} ~&~ \g^0
}\Big)\quad.
\end{flalign}
We call such objects \textit{finite-dimensional $N$-term $L_\infty$-algebras}.
By a simple degree counting argument, one observes that in this case 
the $L_\infty$-brackets $\ell^\g_n=0$ must necessarily vanish for all $n> N+1$. 
The formal moduli problem of flat $(\g,\ell^\g)$-connections
on a manifold $M$ is then modeled by the $L_\infty$-algebra
$(\g,\ell^\g) \otimes \Omega^{\bullet}(M)$ which is given by the tensor product 
from Proposition \ref{prop:tensorproduct} of the structure $L_\infty$-algebra $(\g,\ell^\g)$
with the de Rham (commutative dg-)algebra $\Omega^{\bullet}(M)$ on $M$.
Our topological-holomorphic higher Chern-Simons theories introduced
in Section \ref{sec:higherCS} will be defined by a similar tensor product 
construction where, in order to obtain their characteristic topological-holomorphic features, 
the de Rham algebra $\Omega^{\bullet}(M)$ will be substituted by the de Rham-Dolbeault
algebra $\Omega^\bullet(M)\widehat{\otimes} \Omega^{0,\bullet}(C)$ of a product
manifold $X := M\times C$ with $C=\CP$ the Riemann sphere. 
\end{rem}

The concept of an $\ad$-invariant symmetric non-degenerate pairing $\h\otimes \h \to \bbC$
on an ordinary Lie algebra $\h$ admits a vast generalization
in the context of $L_\infty$-algebras $(L,\ell)$, where such maps are allowed 
to carry a non-trivial cohomological degree $k\in\bbZ$. The geometric interpretation of such data
is that of $(k+2)$-shifted symplectic structures on the associated formal moduli problem,
and their relevance in physics is that they encode in the case of 
$k=-3$ an action functional and the BV antibracket for the theory modeled by $(L,\ell)$.
\begin{defi}\label{def:cyclic}
A \textit{cyclic structure of degree $k\in\bbZ$} on an $L_\infty$-algebra
$(L,\ell)$ is a non-degenerate and graded symmetric linear map
$\pair{\cdot}{\cdot}:L^{\otimes 2}\to \bbC$ of degree $k\in\bbZ$
which is invariant under the $L_\infty$-structure in the sense that
\begin{flalign}\label{eqn:cyclicproperties}
\xymatrix{
\pair{\cdot}{\cdot}\circ \big(\ell_n\otimes\id\big)\,:\,L^{\otimes(n+1)} \ar[r] 
~&~\bbC
}
\end{flalign}
is graded antisymmetric, for all $n\geq 1$.
\end{defi}

Let us note the following interplay between cyclic structures 
and the tensor product construction from Proposition \ref{prop:tensorproduct}.
\begin{propo}\label{prop:cyclictensor}
Let $(L,\ell)$ be an $L_\infty$-algebra
endowed with a cyclic structure $\pair{\cdot}{\cdot}: L^{\otimes 2}\to \bbC$ 
of degree $k\in\bbZ$ and let $A$ be a commutative dg-algebra with differential denoted by $\dd:A\to A$
and multiplication by $\mu :A^{\otimes 2} \to A$. Then the degree $k$ linear map
\begin{flalign}\label{eqn:cyclictensordef}
\xymatrix@C=4em{
\pair{\cdot}{\cdot}^{\otimes} \,:\,\big(L\otimes A\big)^{\otimes 2} \ar[r]^-{\id\otimes\gamma\otimes\id}
~&~L\otimes L\otimes A\otimes A \ar[r]^-{\pair{\cdot}{\cdot}\otimes\mu}~&~\bbC\otimes A\,\cong\,A
}
\end{flalign}
is graded symmetric and equivariant under the tensor product $L_\infty$-structure $\ell^{\otimes}$ from 
Proposition \ref{prop:tensorproduct} in the sense that
\begin{subequations}\label{eqn:cyclictensorproperties}
\begin{flalign}\label{eqn:cyclictensorproperties1}
\pair{\cdot}{\cdot}^{\otimes}\circ\big(\ell_1^\otimes\otimes  \id + \id\otimes\ell_1^\otimes\big) \,=\,(-1)^k~\dd \circ
\pair{\cdot}{\cdot}^{\otimes}
\end{flalign}
and 
\begin{flalign}\label{eqn:cyclictensorproperties2}
\xymatrix{
\pair{\cdot}{\cdot}^{\otimes} \circ \big(\ell^\otimes_n\otimes\id\big)\,:\,
\big(L\otimes A\big)^{\otimes(n+1)}\ar[r]~&~A
}
\end{flalign}
\end{subequations}
is graded antisymmetric, for all $n\geq 2$.
\end{propo}
\begin{proof}
The second property \eqref{eqn:cyclictensorproperties2} is a direct consequence
of the definitions of $\pair{\cdot}{\cdot}^{\otimes}$ and $\ell^\otimes$.
To verify the first property \eqref{eqn:cyclictensorproperties1}, it suffices to 
evaluate on elements $(x\otimes a)\otimes (y\otimes b)\in (L\otimes A)^{\otimes 2}$ 
with all factors of homogeneous degree. From the definitions of $\pair{\cdot}{\cdot}^{\otimes}$ and $\ell^\otimes$
one directly computes
\begin{flalign}
\nn & \pair{\cdot}{\cdot}^{\otimes}\circ\big(\ell_1^\otimes\otimes  \id + \id\otimes\ell_1^\otimes\big)
\big((x\otimes a)\otimes (y\otimes b)\big)\\
\nn&\qquad\quad \,=\, (-1)^{\vert a\vert\,\vert y\vert}\,\Big(
\pair{\ell_1 x}{y}\,ab + (-1)^{\vert x\vert}\,\pair{x}{\ell_1 y}\,ab
+ (-1)^{\vert x\vert+\vert y\vert}\,\pair{x}{y}\,\dd (a b)\Big)\\
&\qquad\quad\,=\,(-1)^k\,\dd\circ \pair{\cdot}{\cdot}^{\otimes}\big((x\otimes a)\otimes (y\otimes b)\big)\quad,
\end{flalign}
where in the last step we used that the first two terms cancel as a consequence
of the properties \eqref{eqn:cyclicproperties} of a cyclic structure
and that $\pair{x}{y}$ is of degree $k+\vert x\vert+\vert y\vert$.
\end{proof}

\begin{rem}\label{rem:cyclictensor}
In analogy to Remark \ref{rem:tensorproduct},
we will sometimes denote, with a slight abuse of notation, 
the degree $k$ linear map $\pair{\cdot}{\cdot}^{\otimes}$
from Proposition \ref{prop:cyclictensor} simply by $\pair{\cdot}{\cdot}$.
\end{rem}

%%%%%%%%%%%%%%%%%%%%%%%%%%%%%%%%%%%%%%%%%%%%%%%%
%%%%%%%%%%%%%%%%%%%%%%%%%%%%%%%%%%%%%%%%%%%%%%%%

\section{\label{sec:higherCS}Topological-holomorphic higher Chern-Simons theories}
In this section we present a family of $L_\infty$-algebras 
which describes higher-dimensional and higher-categorical 
variants of the topological-holomorphic Chern-Simons theory
of Costello and Yamazaki \cite{CY3}. Our family of 
theories is labeled in particular by a positive integer
$d\in\bbZ_{\geq 1}$, which describes both the categorical depth
of the higher topological-holomorphic Chern-Simons theory and the number of spatial dimensions
of its associated $(d+1)$-dimensional integrable field theories.
The case of $d=1$ corresponds to the usual $4$-dimensional
topological-holomorphic Chern-Simons theory
of Costello and Yamazaki \cite{CY3}.
\sk

For a fixed positive integer $d\in\bbZ_{\geq 1}$, the type of manifolds
on which we define our topological-holomorphic higher Chern-Simons theories
are products
\begin{flalign}
X\,:=\,M\times C\,:=\,M\times \CP
\end{flalign}
of a $(d+1)$-dimensional manifold $M$ and the Riemann 
sphere $C=\CP$.\footnote{All constructions and results
of this paper can be generalized to the case of higher-genus 
Riemann surfaces $C$, but this requires some additional ingredients 
and techniques which render the formalism heavier. The details for 
the case of ordinary $4d$ topological-holomorphic Chern-Simons theory 
are spelled out explicitly in \cite[Section 4]{BSV2}, and the same methods
clearly apply to higher dimensions too.}
We denote by
\begin{flalign}\label{eqn:dRDol}
\Omega^{\bullet,(0,\bullet)}(X)\,:=\,\Omega^\bullet(M)\,\widehat{\otimes}\, \Omega^{0,\bullet}(C)
\end{flalign}
the commutative dg-algebra which is given by the (completed projective) tensor product
of the de Rham algebra $\Omega^\bullet(M)$ on $M$ and the Dolbeault algebra $\Omega^{0,\bullet}(C)$
on $C$ (with respect to their standard Fr{\'e}chet topologies). Note that, as a $\bbZ$-graded vector space,
the de Rham-Dolbeault algebra can be decomposed as a direct sum
\begin{flalign}\label{eqn:dRDolsplit}
\Omega^{\bullet,(0,\bullet)}(X)\,=\,\Omega^{\bullet,(0,0)}(X) \oplus \Omega^{\bullet-1,(0,1)}(X)\quad,
\end{flalign}
where $\Omega^{\bullet,(0,0)}(X)$ contains the forms on $X$ whose legs are only along $M$, while
$\Omega^{\bullet-1,(0,1)}(X)$ contains the forms on $X$ with precisely $1$ leg along $C$ 
and all others along $M$.
\begin{defi}\label{def:CSauxiliary}
Let $(\g,\ell^\g)$ be a finite-dimensional $d$-term $L_\infty$-algebra, see Remark \ref{rem:Linftyalgebra}.
The $L_\infty$-algebra of \textit{topological-holomorphic
$(\g,\ell^\g)$-Chern-Simons theory without singularities and boundary conditions} 
on the $(d+3)$-dimensional manifold $X=M\times C$ is defined as 
the tensor product
\begin{flalign}
\big(\E(X),\ell\big)\,:=\, (\g,\ell^\g)\otimes \Omega^{\bullet,(0,\bullet)}(X)
\end{flalign}
from Proposition \ref{prop:tensorproduct} of the structure $L_\infty$-algebra $(\g,\ell^\g)$ with the 
de Rham-Dolbeault commutative dg-algebra \eqref{eqn:dRDol}.
\end{defi}

\begin{rem}\label{rem:CSauxiliary}
To obtain an understanding of the kind of objects described by
the $L_\infty$-algebra $\big(\E(X),\ell\big)$ from Definition \ref{def:CSauxiliary},
let us spell out explicitly its Maurer-Cartan elements \eqref{eqn:MCequation} and (higher) 
gauge transformations \eqref{eqn:MCgauge}. For this it is convenient to make use of the direct
sum decomposition \eqref{eqn:dRDolsplit} in order to decompose the underlying
$\bbZ$-graded vector space of this $L_\infty$-algebra according to
\begin{flalign}
\E(X)\,=\,\E(X)_{\|}^{}\oplus \E(X)_{\perp}^{}\,:=\,\Big(\g\otimes \Omega^{\bullet,(0,0)}(X)\Big) \oplus
\Big(\g\otimes \Omega^{\bullet-1,(0,1)}(X)\Big)\quad.
\end{flalign}
Spelling this out in more detail, this means that the components
\begin{flalign}
\E(X)_{\|}^j\,=\,\bigoplus_{k=0}^{d+1} \Big(\g^{j-k}\otimes \Omega^{k,(0,0)}(X)\Big)\quad,
\qquad
\E(X)_{\perp}^j\,=\,\bigoplus_{k=0}^{d+1} \Big(\g^{j-k-1}\otimes \Omega^{k,(0,1)}(X)\Big)\quad,
\end{flalign} 
for $j\in\bbZ$, are described by forms of various degrees on $X$ taking values in the appropriate components of $\g$.
The characteristic feature of $\E(X)_{\|}$ is that it consists of forms whose legs are only along $M$,
while $\E(X)_{\perp}$ consists of forms with precisely $1$ leg along $C$ and all others along $M$. 
A Maurer-Cartan element in $\big(\E(X),\ell\big)$ is then given by a pair of degree $1$ 
elements $A\oplus \xi\in \E(X)_{\|}^1\oplus \E(X)_{\perp}^1$ satisfying the differential equations
\begin{subequations}
\begin{flalign}
\dd_M A + \sum_{n\geq 1}\tfrac{1}{n!}\,\ell_n^{\g}\big(A^{\otimes n}\big)\,&=\,0\quad,\\
\delbar A + \dd_M\xi + \sum_{n\geq 0}\tfrac{1}{n!}\, \ell_{n+1}^\g\big(A^{\otimes n}\otimes \xi\big)\,&=\,0\quad,
\end{flalign}
\end{subequations}
where $\dd_M$ is the de Rham differential on $M$, $\delbar$ is the Dolbeault differential on $C$
and $\ell_n^\g$ denotes the extension of the $L_\infty$-brackets on $\g$
to $\g$-valued forms (see also Remark \ref{rem:tensorproduct}). 
Observe that the first equation expresses flatness 
of the higher connection $A$ along $M$, while the second equation expresses that $A$ is holomorphic
along $C$, up to a violation controlled by $\xi$.
\sk

Gauge and higher gauge transformations take in this example the following form:
A $1$-gauge transformation is given by a pair of degree $0$ elements $\epsilon_{(0)}\oplus\chi_{(0)}
\in \E(X)_{\|}^0\oplus \E(X)_{\perp}^0$ which acts on Maurer-Cartan elements 
$A\oplus\xi\in \E(X)^{1}_{\|} \oplus \E(X)^{1}_{\perp}$ according to
\begin{subequations}
\begin{flalign}
\delta_{\epsilon_{(0)}\oplus \chi_{(0)}}A \,&=\, \dd_M \epsilon_{(0)} + 
\sum_{n\geq 0}\tfrac{1}{n!}\,\ell_{n+1}^\g\big(A^{\otimes n}\otimes \epsilon_{(0)}\big)\quad,\\
\delta_{\epsilon_{(0)}\oplus \chi_{(0)}}\xi\,&=\, 
\delbar \epsilon_{(0)} + \dd_M \chi_{(0)} + 
\sum_{n\geq 0}\tfrac{1}{n!}\,\ell_{n+1}^\g\big(A^{\otimes n}\otimes \chi_{(0)}\big)
+ \sum_{n\geq 0}\tfrac{1}{n!}\,\ell_{n+2}^\g\big(A^{\otimes n}\otimes \xi\otimes \epsilon_{(0)}\big)\quad.
\end{flalign}
\end{subequations}
For $2\leq k\leq d$, a $k$-gauge transformation is given by a degree 
$1-k$ element $\epsilon_{(1-k)}\oplus \chi_{(1-k)}\in 
\E(X)^{1-k}_{\|} \oplus \E(X)^{1-k}_{\perp}$ and 
it acts on $(k-1)$-gauge transformations $\epsilon_{(2-k)}\oplus \chi_{(2-k)}\in 
\E(X)^{2-k}_{\|} \oplus \E(X)^{2-k}_{\perp}$ according to
\begin{subequations}
\begin{flalign}
\delta_{\epsilon_{(1-k)}\oplus \chi_{(1-k)}}\epsilon_{(2-k)} \,&=\, 
\dd_M \epsilon_{(1-k)} + \sum_{n\geq 0}\tfrac{1}{n!}\,\ell_{n+1}^\g\big(A^{\otimes n}\otimes \epsilon_{(1-k)}\big)\quad,\\
\nn \delta_{\epsilon_{(1-k)}\oplus \chi_{(1-k)}}\chi_{(2-k)}\,&=\, \delbar \epsilon_{(1-k)} 
+ \dd_M \chi_{(1-k)} + \sum_{n\geq 0}\tfrac{1}{n!}\,\ell_{n+1}^\g\big(A^{\otimes n}\otimes \chi_{(1-k)}\big)\\
&\qquad \qquad+ \sum_{n\geq 0}\tfrac{1}{n!}\,\ell_{n+2}^\g\big(A^{\otimes n}\otimes \xi \otimes \epsilon_{(1-k)}\big)\quad.
\end{flalign}
\end{subequations}
Since $\E(X)^j=0$, for all $j \leq -d$, all $(k\geq d+1)$-gauge transformations are trivial in this example.
\end{rem}

\begin{rem}\label{rem:CSauxiliary2}
It is important to highlight that in Definition \ref{def:CSauxiliary} 
the positive integer $d\in\bbZ_{\geq 1}$ enters in two a priori
completely unrelated ways: 1.)~It fixes the dimension of the $(d+1)$-dimensional
manifold $M$ and 2.)~it determines the categorical depth of the structure $d$-term $L_\infty$-algebra $(\g,\ell^\g)$.
While this link between dimension and categorical depth is inessential for the observations 
in Remark \ref{rem:CSauxiliary}, it becomes crucial when discussing
degree $-3$ cyclic structures on such types of $L_\infty$-algebras, which 
provide in particular action functionals and BV antibrackets for such theories.
Even though a precise construction of such degree $-3$ cyclic structures becomes only possible
after one has implemented suitable singularities and boundary 
conditions\footnote{This is completely analogous to the observation that 
the action functional $S_\omega(\A)= \int_X \omega\wedge\mathrm{CS}(\A)$
of Costello and Yamazaki \cite{CY3} becomes non-degenerate and gauge invariant
only once the fields $\A$ are subject to suitable singularities and boundary conditions, see e.g.\ \cite{BSV}
and also \cite{BSV2} for further details.}
at the zeros and poles of a chosen meromorphic $1$-form $\omega$ on $C$,
see Section \ref{sec:higherIFT} for the details, we can already explain now
the origin of this link. Any cyclic structure of degree $k\in\bbZ$
on the structure $d$-term $L_\infty$-algebra $(\g,\ell^\g)$ is by Definition
\ref{def:cyclic} non-degenerate, hence it determines
in particular an isomorphism $\g\stackrel{\cong}{\to} \g^\ast[k]$
between the underlying graded vector space $\g$ and its $k$-shifted
dual $\g^\ast[k]$. Since $\g$ is concentrated in degrees $\{1-d,\dots,0\}$
and $\g^\ast[k]$ in degrees $\{-k,1-k,\dots,d-1-k\}$, one observes that,
under the assumption that the extremal components $\g^0\neq 0$
and $\g^{1-d}\neq 0$ are non-trivial, such an isomorphism may only exist for $k=d-1$.
Given now any cyclic structure $\pair{\cdot}{\cdot}^\g:\g^{\otimes 2}\to \bbC$
of degree $d-1$, we obtain from the construction in Proposition \ref{prop:cyclictensor}
a degree $d-1$ linear map
\begin{flalign}
\pair{\cdot}{\cdot}^\g \,:\, \E(X)^{\otimes 2}~\longrightarrow~\Omega^{\bullet,(0,\bullet)}(X)
\end{flalign}
describing a kind of `densitized' cyclic structure on $\E(X)$.
Upon implementing suitable singularities and boundary conditions
as in Section \ref{sec:higherIFT}, this map takes 
values in $\Omega^{\bullet,(0,\bullet)}\big(X,L_{(\omega)}\big)$,
where $L_{(\omega)}\to C$ denotes the holomorphic line bundle associated
with the divisor of the meromorphic $1$-form $\omega$.
Restricting to compactly supported sections and composing with the degree $-(d+2)$
integration map $\int_X \omega\wedge(\,\cdot\,) : \Omega_\cc^{\bullet,(0,\bullet)}\big(X,L_{(\omega)}\big)\to \bbC$ 
over the $(d+3)$-dimensional product manifold $X = M\times C$ then defines a cyclic structure
of degree $-3$.
\end{rem}

%%%%%%%%%%%%%%%%%%%%%%%%%%%%%%%%%%%%%%%%%%%%%%%%
%%%%%%%%%%%%%%%%%%%%%%%%%%%%%%%%%%%%%%%%%%%%%%%%

\section{\label{sec:higherLax}Higher Lax connections}
In this section we demonstrate how higher Lax connections 
arise naturally in our framework by allowing the fields 
described by the $L_\infty$-algebra $\big(\E(X),\ell\big)$ 
from Definition \ref{def:CSauxiliary} to develop `suitable'
singularities, which we formalize by non-negative divisors on $C=\CP$.
Such singularities are allowed to occur only at the zeros
of a fixed meromorphic $1$-form $\omega$ on $C$ and they 
are constrained by algebraic properties involving
the structure $L_\infty$-algebra $(\g,\ell^\g)$ and
a choice of cyclic structure $\pair{\cdot}{\cdot}^\g : 
\g^{\otimes 2}\to \bbC$ of degree $d-1$.
This will guarantee in particular that the fields with singularities form an $L_\infty$-algebra.
Throughout this section we make arbitrary
but fixed choices for the data $\pair{\cdot}{\cdot}^\g$ and $\omega$,
and we write
\begin{flalign}\label{eqn:omegasplit}
(\omega)\,=\,(\omega)_0 + (\omega)_{\infty}
\end{flalign}
for the decomposition of the divisor $(\omega): C\to \bbZ$
of $\omega$ into its non-negative part $(\omega)_0\geq 0$ (describing
the zeros of $\omega$) and its non-positive part $(\omega)_\infty\leq 0$ 
(describing the poles of $\omega$).
\sk

To provide a precise definition of what we mean
by `suitable' singularities, we require access to the individual
form components along the $(d+1)$-dimensional manifold $M$.
This is most easily achieved by assuming that $M$ is parallelizable
and choosing any basis for the $C^\infty(M)$-module $\Omega^1(M)$ of $1$-forms. While there
might exist a more general framework based on local trivializations,
we will not explore this here in order to avoid additional technicalities.
Therefore, we make the following assumption throughout the rest of this paper.
\begin{assu}\label{assu:parallelizable}
The $(d+1)$-dimensional manifold $M$ is parallelizable.
\end{assu}

Fixing any $C^\infty(M)$-module basis
\begin{flalign}\label{eqn:formbasis}
\big\{e^1,e^2,\dots,e^{d+1}\in\Omega^1(M)\big\}\quad,
\end{flalign}
we can define for any $(d+1)$-tuple $\und{a}:=(a_1,\dots,a_{d+1})\in \{0,1\}^{d+1}\subseteq \bbN^{d+1}$
a differential form
\begin{flalign}
e^{\und{a}}\,:=\, (e^1)^{a_1}\wedge (e^2)^{a_2}\wedge\cdots\wedge (e^{d+1})^{a_{d+1}}\,\in\,\Omega^{\bullet}(M)
\end{flalign}
of degree $\vert \und{a}\vert := \sum_{k=1}^{d+1}a_k$, where by definition $(e^k)^0 := 1$ is the constant
function and $(e^k)^1:= e^k$.
This defines a direct sum decomposition
\begin{flalign}\label{eqn:formdecomposition}
\Omega^\bullet(M)\,=\,\!\!\bigoplus_{\und{a}\in \{0,1\}^{d+1}}\!\!\Omega^{\und{a}}(M)
\end{flalign}
of the $C^\infty(M)$-module $\Omega^\bullet(M)$ into rank $1$ free modules
$\Omega^{\und{a}}(M):= C^\infty(M)\,e^{\und{a}}$.
\sk

Note that this direct sum decomposition interplays nicely with the 
$\wedge$-product and the de Rham differential of forms.
Given any $\und{a},\und{b}\in \{0,1\}^{d+1}$ such that 
$\und{a} + \und{b} \in \{0,1\}^{d+1}$, the $\wedge$-product restricts to a map
\begin{flalign}
\wedge\,:\, \Omega^{\und{a}}(M)\otimes \Omega^{\und{b}}(M)~\longrightarrow~
\Omega^{\und{a}+\und{b}}(M)\quad,
\end{flalign}
and it is trivial in the case where $\und{a} + \und{b} \not\in \{0,1\}^{d+1}$ since $e^k\wedge e^k=0$.
Given any $\und{a}\in \{0,1\}^{d+1}$ and introducing its subset of successors
\begin{flalign}\label{eqn:successor}
\mathsf{suc}(\und{a})\,:=\,\Big\{\und{b}\in \{0,1\}^{d+1}~~\Big\vert~~\und{b}-\und{a}\in \{0,1\}^{d+1}~\text{and}~
\vert \und{b}-\und{a}\vert =1\Big\}\,\subseteq\,\{0,1\}^{d+1}\quad,
\end{flalign}
the de Rham differential restricts to a map
\begin{flalign}
\dd_M\,:\,  \Omega^{\und{a}}(M)~\longrightarrow~\!\!\bigoplus_{\und{b}\in \mathsf{suc}(\und{a})}\!\!\Omega^{\und{b}}(M)\quad.
\end{flalign}
\begin{defi}\label{def:singularity}
Let $(\g,\ell^\g)$ be a finite-dimensional $d$-term $L_\infty$-algebra
endowed with a cyclic structure $\pair{\cdot}{\cdot}^\g$ of degree $d-1$,
and let $\omega$ be a meromorphic $1$-form on $C$.
A \textit{singularity structure} for the topological-holomorphic $(\g,\ell^\g)$-Chern-Simons
theory from Definition \ref{def:CSauxiliary} consists of a family of non-negative divisors on $C$
\begin{flalign}
\mathscr{D}\,:=\,\Big\{ D^i_{\und{a}}\geq 0~~\Big\vert~~ i\in\{1-d,\dots,0\}~,~~\und{a}\in\{0,1\}^{d+1}\Big\}\quad,
\end{flalign}
which satisfies the following properties:
\begin{itemize}
\item[(i)] Given any $i\in\{1-d,\dots,0\}$ and $\und{a}\in \{0,1\}^{d+1}$, then
\begin{flalign}
D_{\und{a}}^i\,\leq \,D_{\und{b}}^i\quad,
\end{flalign}
for all successors $\und{b}\in \mathsf{suc}(\und{a})$ from \eqref{eqn:successor}.

\item[(ii)] Given any $n\geq 1$ and $i_1,\dots,i_n\in \{1-d,\dots,0\}$ such that
$\ell^{\mathfrak{g}}_n : \bigotimes_{k=1}^n\g^{i_k}\to \g^{2-n + \sum_{k=1}^n i_k}$ is not the zero map,
then
\begin{flalign}
\sum_{k=1}^n D_{\und{a}_k}^{i_k} \,\leq \, D^{2-n + \sum_{k=1}^n i_k}_{\sum_{k=1}^n\und{a}_k}\quad,
\end{flalign}
for all $\und{a}_1,\dots,\und{a}_n\in\{0,1\}^{d+1}$ with $\sum_{k=1}^n\und{a}_k \in\{0,1\}^{d+1}$.

\item[(iii)] Given any $i\in \{1-d,\dots,0\}$ and $\und{a}\in\{0,1\}^{d+1}$, then
\begin{flalign}
D^i_{\und{a}} + D^{1-d-i}_{\und{1}-\und{a}}\,=\,(\omega)_0
\end{flalign}
yields the non-negative part $(\omega)_0$ of the divisor of $\omega$,
where $\und{1}:= (1,\dots,1)\in \{0,1\}^{d+1}$ denotes the tuple consisting only of $1$s.
\end{itemize}
\end{defi}
 
Given any singularity structure $\mathscr{D}$ as 
in Definition \ref{def:singularity}, one can implement its prescribed 
singularities in the topological-holomorphic $(\g,\ell^\g)$-Chern-Simons
theory $\big(\E(X),\ell\big)$ from Definition \ref{def:CSauxiliary} by the following construction:
Making use of \eqref{eqn:formdecomposition}, we obtain a direct sum decomposition
\begin{flalign}
\E(X)\,=\, \!\!\bigoplus_{i\in \{1-d,\dots,0\}}~\bigoplus_{\und{a}\in\{0,1\}^{d+1}} \g^{i}\otimes \Big(\Omega^{\und{a}}(M)\,\widehat{\otimes}\,\Omega^{0,\bullet}(C)\Big)\quad.
\end{flalign}
We then define
\begin{flalign}\label{eqn:singularitytwisted}
\E_\mathscr{D}^{}(X)\,:=\, \!\!\bigoplus_{i\in \{1-d,\dots,0\}}~\bigoplus_{\und{a}\in\{0,1\}^{d+1}} 
\g^{i}\otimes \Big(\Omega^{\und{a}}(M)\,\widehat{\otimes}\,\Omega^{0,\bullet}\big(C,L_{D^i_{\und{a}}}\big)\Big)
\end{flalign}
by twisting the component-wise Dolbeault complexes with the holomorphic line bundles
$L_{D^i_{\und{a}}}\to C$ associated with the corresponding members of the family of divisors 
$\mathscr{D} = \{D^i_{\und{a}}\}$. Note that $\E_\mathscr{D}^{}(X)$ inherits the structure of 
a $\bbZ$-graded vector space. Using that $L_{D}\otimes L_{D^\prime}\cong L_{D+D^\prime}$,
for any two divisors $D,D^\prime$, and employing the holomorphic line bundle morphisms
$L_D\to L_{D^\prime}$ associated with partially ordered divisors $D\leq D^\prime$, 
one immediately recognizes that the properties (i) and (ii) of Definition \ref{def:singularity}
are designed precisely in such a way that $\E_\mathscr{D}^{}(X)$ can be endowed with the de Rham differential
$\dd_M$ and the extension from Remark \ref{rem:tensorproduct} of the
$L_\infty$-structure $\ell^\g$. (Property (iii) of Definition \ref{def:singularity} is not relevant
at the present point, but it will play an important role later when we discuss cyclic structures.)
This leads us to the following definition.
\begin{defi}\label{def:CSsingularity}
Let $(\g,\ell^\g)$ be a finite-dimensional $d$-term $L_\infty$-algebra
endowed with a cyclic structure $\pair{\cdot}{\cdot}^\g$ of degree $d-1$,
and let $\omega$ be a meromorphic $1$-form on $C$. Choose any 
singularity structure $\mathscr{D}$ as in Definition \ref{def:singularity}.
The $L_\infty$-algebra $\big(\E_{\mathscr{D}}^{}(X),\ell\big)$ 
of \textit{topological-holomorphic $(\g,\ell^\g)$-Chern-Simons theory with singularities
$\mathscr{D}$ but no boundary conditions} on the $(d+3)$-dimensional manifold $X=M\times C$ 
is defined by the $\bbZ$-graded vector space \eqref{eqn:singularitytwisted} and 
the $L_\infty$-structure given by $\ell_1 := \dd_M + \delbar + \ell_1^\g$
and $\ell_n := \ell_n ^\g$, for all $n\geq 2$.
\end{defi}

To identify the theory of higher Lax connections associated
with the $L_\infty$-algebra $\big(\E_{\mathscr{D}}^{}(X),\ell\big)$, we use the homological
perturbation techniques from \cite{BSV2}. For this we observe that, since the divisors $D^i_{\und{a}}\geq 0$
are by definition all non-negative, all divisor-twisted Dolbeault complexes entering \eqref{eqn:singularitytwisted} 
have trivial first cohomology $\mathsf{H}^1\Omega^{0,\bullet}\big(C,L_{D^i_{\und{a}}}\big)\cong 0$.
Their zeroth cohomology $\mathsf{H}^0\Omega^{0,\bullet}\big(C,L_{D^i_{\und{a}}}\big)\cong \O_{D^i_{\und{a}}}(C)$
is given by the (finite-dimensional) vector spaces of divisor-conditioned meromorphic functions on $C$.
Choosing now any family of continuous strong deformation retracts (e.g.\ from Hodge theory)
\begin{equation}
\begin{tikzcd}
\big(\O_{D^i_{\und{a}}}(C),0\big) \ar[rr,shift right=-1ex,"i^i_{\und{a}}"] && \ar[ll,shift right=-1ex,"p^i_{\und{a}}"] 
\big(\Omega^{0,\bullet}(C,L_{D^i_{\und{a}}}),\delbar \big) \ar[loop,out=-20,in=18,distance=30,swap,"h^i_{\und{a}}"]
\end{tikzcd}\quad,
\end{equation} 
for all $i\in\{1-d,\dots,0\}$ and $\und{a}\in\{0,1\}^{d+1}$, we obtain via the 
totalization construction in \cite[Appendix A.1]{BSV2} a strong deformation retract
\begin{equation}\label{eqn:defretsingular}
\begin{tikzcd}[column sep=large]
\big(\L(M),\ell_1^\prime\big) \ar[rr,shift right=-1ex,"i"] && \ar[ll,shift right=-1ex,"\widetilde{p}\,:=\,p + p(\dd_M+\ell_1^\g)h"] 
\big(\E_\mathscr{D}^{}(X),\ell_1 \big) \ar[loop,out=-25,in=25,distance=30,swap,"h"]
\end{tikzcd}
\end{equation} 
from the underlying cochain complex of the $L_\infty$-algebra from Definition \ref{def:CSsingularity}
to the cochain complex defined by 
\begin{subequations}\label{eqn:singularM}
\begin{flalign}
\L(M)\,=\,
\!\!\bigoplus_{i\in \{1-d,\dots,0\}}~\bigoplus_{\und{a}\in\{0,1\}^{d+1}} 
\g^{i}\otimes \Omega^{\und{a}}(M) \otimes\O_{D^i_{\und{a}}}(C)\quad,
\end{flalign}
endowed with the $\bbZ$-graded vector space structure inherited from $\g$ and $\Omega^\bullet(M)$, 
and differential 
\begin{flalign}
\ell_1^\prime \,:=\, \dd_M +\ell_1^\g
\end{flalign}
\end{subequations}
given by the sum of the de Rham differential and the one on the structure $L_\infty$-algebra. 
Property (ii) of Definition \ref{def:singularity} further implies that the higher $L_\infty$-brackets
$\ell^\g_n$, for $n\geq 2$, extend in the sense of Remark \ref{rem:tensorproduct} to 
$\big(\L(M),\ell_1^\prime\big)$, and thereby define an $L_\infty$-algebra
$\big(\L(M),\ell^\prime\big)$. 
The cochain map $i$ from \eqref{eqn:defretsingular} is by construction a quasi-isomorphism
and it is furthermore compatible with the $L_\infty$-structures.
The following result then follows immediately.
\begin{propo}\label{prop:higherLax}
The cochain map $i$ from \eqref{eqn:defretsingular} defines a (strict) morphism
of $L_\infty$-algebras
\begin{flalign}
\xymatrix{
i\,:\,\big(\L(M),\ell^\prime\big) \ar[r]^-{\sim}~&~\big(\E_\mathscr{D}^{}(X),\ell\big)
}
\end{flalign}
which is further a weak equivalence between the $L_\infty$-algebra 
$\big(\E_\mathscr{D}^{}(X),\ell\big)$ from Definition \ref{def:CSsingularity} and
the $L_\infty$-algebra $\big(\L(M),\ell^\prime\big)$ defined
by the cochain complex \eqref{eqn:singularM} and the $L_\infty$-brackets $\ell^\prime_n = \ell^\g_n$, for all $n\geq 2$.
\end{propo}

\begin{rem}\label{rem:higherLax}
We announced above that the $L_\infty$-algebra $\big(\L(M),\ell^\prime\big)$
describes higher Lax connections, so let us now substantiate this claim.
Using the inclusions $\O_{D^{i}_{\und{a}}}(C)\subseteq \M(C)$ of the divisor-conditioned
meromorphic functions into the algebra of all meromorphic functions on $C$, 
we can identify Maurer-Cartan elements in the $L_\infty$-algebra $\big(\L(M),\ell^\prime\big)$ 
with degree $1$ elements $A\in \big(\g\otimes \big(\Omega^\bullet(M)\widehat{\otimes}\M(C)\big)\big)^1$ 
that satisfy the flatness condition for $(\g,\ell^\g)$-connections
\begin{flalign}
\dd_M A + \sum_{n\geq 1}\tfrac{1}{n!}\,\ell_n^\g\big(A^{\otimes n}\big)\,=\,0
\end{flalign}
along $M$ and are meromorphic along $C$, with poles controlled by the singularity structure $\mathscr{D}$.
These are precisely the desired features of a higher Lax connection for the structure
$d$-term $L_\infty$-algebra $(\g,\ell^\g)$. The gauge and higher gauge
transformations of such higher Lax connections are easily described by specializing 
their general description from Remark \ref{rem:Linftyalgebra}. Explicitly,
\begin{subequations}
\begin{flalign}
\delta_{\epsilon_{(0)}}A\,&=\,\dd_M\epsilon_{(0)} + 
\sum_{n\geq 0}\tfrac{1}{n!}\,\ell_{n+1}^\g\big(A^{\otimes n}\otimes \epsilon_{(0)}\big)\quad,\\
\delta_{\epsilon_{(1-k)}}\epsilon_{(2-k)}\,&=\,\dd_M \epsilon_{(1-k)} + 
\sum_{n\geq 0}\tfrac{1}{n!}\,\ell_{n+1}^\g\big(A^{\otimes n}\otimes \epsilon_{(1-k)}\big)\quad,
\end{flalign}
\end{subequations}
for all Maurer-Cartan elements $A\in \L(M)^1$ 
and all $\epsilon_{(0)}\in \L(M)^0$, $\epsilon_{(2-k)}\in \L(M)^{2-k}$ and
$\epsilon_{(1-k)}\in \L(M)^{1-k}$, with $k\geq 2$. Note that from the 
definition of $\L(M)$ in \eqref{eqn:singularM} it follows that also
the $k$-gauge transformations $\epsilon_{(1-k)}\in \L(M)^{1-k}$, for $k\geq 1$, 
are meromorphic along $C$, with poles controlled by the singularity structure $\mathscr{D}$.
\sk

It is worthwhile to highlight that there exists an interesting
notion of higher parallel transport for higher flat connections \cite{AS},
taking the form of a morphism of differential graded coalgebras
\begin{flalign}
\mathrm{hol}_A^{\infty}\,:\,C_\bullet(M) ~\longrightarrow~\hat{\mathsf{B}}\big(\hat{\mathbb{U}}(\g)\widehat{\otimes}\M(C)\big)
\quad,
\end{flalign}
for every Maurer-Cartan element $A\in \L(M)^1$.
The domain of this morphism is the differential graded coalgebra of singular chains in $M$ and its
codomain is given by the completed bar construction $\hat{\mathsf{B}}$
of the completed tensor product of the completed 
universal enveloping algebra $\hat{\mathbb{U}}(\g)$ of the structure $L_\infty$-algebra
$(\g,\ell^\g)$ with the algebra $\M(C)$ of meromorphic functions on $C$.
In a language more tailored to integrable field theorists, this means that one can assign
to every higher Lax connection $A\in \L(M)^1$ a family of conserved charges taking values 
in the completed universal enveloping algebra $\hat{\mathbb{U}}(\g)$, which depend meromorphically on $C$
and are localized on non-trivial $k$-cycles $Z_k(M)\subseteq C_k(M)$ in spacetime $M$. This covers a wide range of
shapes such as non-contractible spheres $\mathbb{S}^k\subseteq M$ or tori $\mathbb{T}^k\subseteq M$
of any dimension $k\leq \dim(M) = d+1$.
\end{rem}

%%%%%%%%%%%%%%%%%%%%%%%%%%%%%%%%%%%%%%%%%%%%%%%%
%%%%%%%%%%%%%%%%%%%%%%%%%%%%%%%%%%%%%%%%%%%%%%%%

\section{\label{sec:higherIFT}\texorpdfstring{$(d+1)$}{(d+1)}-dimensional integrable field theories}
In this section we demonstrate how integrable field theories on
the $(d+1)$-dimensional manifold $M$ arise naturally in our framework 
by imposing `suitable' boundary conditions on the topological-holomorphic $(\g,\ell^\g)$-Chern-Simons theory 
$\big(\E_{\mathscr{D}}^{}(X),\ell\big)$ with prescribed singularities $\mathscr{D}$
from Definition \ref{def:CSsingularity}, 
which we formalize by non-positive divisors on $C=\CP$.
Such boundary conditions are imposed only at the poles
of the fixed meromorphic $1$-form $\omega$ on $C$ and they 
are constrained, in complete analogy to the case of singularity structures
from Definition \ref{def:singularity}, 
by algebraic properties involving the structure $L_\infty$-algebra $(\g,\ell^\g)$ and
the chosen cyclic structure $\pair{\cdot}{\cdot}^\g : 
\g^{\otimes 2}\to \bbC$ of degree $d-1$. This will guarantee that
the boundary conditioned fields form an $L_\infty$-algebra, which furthermore
carries a natural cyclic structure of degree $-3$. 
In order to enforce that all components of the higher Lax connection
are fully determined by the fields of the integrable field theory,
our definition of boundary conditions below includes a fourth property
which links them to the chosen singularity structure $\mathscr{D}$ 
(and, in general, also the genus of $C$) and implies that
the divisor-twisted Dolbeault cohomologies containing the degrees of
freedom of the higher Lax connection are trivial.
\begin{defi}\label{def:bdycondition}
Let $\mathscr{D} = \{D^i_{\und{a}}\}$ be a singularity structure in the sense of Definition \ref{def:singularity}.
A \textit{local boundary condition} for the topological-holomorphic $(\g,\ell^\g)$-Chern-Simons
theory with singularities $\mathscr{D}$ from Definition \ref{def:CSsingularity} 
consists of a family of non-positive divisors on $C$
\begin{flalign}
\mathscr{B}\,:=\,\Big\{ B^i_{\und{a}}\leq 0~~\Big\vert~~ i\in\{1-d,\dots,0\}~,~~\und{a}\in\{0,1\}^{d+1}\Big\}\quad,
\end{flalign}
which satisfies the following properties:
\begin{itemize}
\item[(i)] Given any $i\in\{1-d,\dots,0\}$ and $\und{a}\in \{0,1\}^{d+1}$, then
\begin{flalign}
B_{\und{a}}^i\,\leq \,B_{\und{b}}^i\quad,
\end{flalign}
for all successors $\und{b}\in \mathsf{suc}(\und{a})$ from \eqref{eqn:successor}.

\item[(ii)] Given any $n\geq 1$ and $i_1,\dots,i_n\in \{1-d,\dots,0\}$ such that
$\ell^{\mathfrak{g}}_n : \bigotimes_{k=1}^n\g^{i_k}\to \g^{2-n + \sum_{k=1}^n i_k}$ is not the zero map,
then
\begin{flalign}\label{eqn:bdypropii}
\sum_{k=1}^n B_{\und{a}_k}^{i_k} \,\leq \, B^{2-n + \sum_{k=1}^n i_k}_{\sum_{k=1}^n\und{a}_k}\quad,
\end{flalign}
for all $\und{a}_1,\dots,\und{a}_n\in\{0,1\}^{d+1}$ with $\sum_{k=1}^n\und{a}_k \in\{0,1\}^{d+1}$.

\item[(iii)] Given any $i\in \{1-d,\dots,0\}$ and $\und{a}\in\{0,1\}^{d+1}$, then
\begin{flalign}
B^i_{\und{a}} + B^{1-d-i}_{\und{1}-\und{a}}\,=\,(\omega)_\infty
\end{flalign}
yields the non-positive part $(\omega)_\infty$ of the divisor of $\omega$,
where $\und{1}:= (1,\dots,1)\in \{0,1\}^{d+1}$ denotes the tuple consisting of only $1$s.

\item[(iv)] Given $i\in  \{1-d,\dots,0\}$ and $\und{a}\in\{0,1\}^{d+1}$ such that
$i + \vert \und{a} \vert =1$, then
\begin{flalign}\label{physical fields cohomology vanishing condition}
\deg\big(D^i_{\und{a}} + B^{i}_{\und{a}}\big)\,=\, g-1 \,=\,-1\quad, 
\end{flalign}
where $g$ denotes the genus of $C$, i.e.\ $g=0$ in our case of $C=\CP$.
\end{itemize}
\end{defi}

Given any local boundary condition $\mathscr{B}$ as in Definition \ref{def:bdycondition},
one can enforce it on the topological-holomorphic $(\g,\ell^\g)$-Chern-Simons
theory $\big(\E_{\mathscr{D}}(X),\ell\big)$ with singularities $\mathscr{D}$ 
from Definition \ref{def:CSsingularity} by the following construction:
Starting from the decomposition in \eqref{eqn:singularitytwisted}, we define
\begin{flalign}\label{eqn:bdytwisted}
\E_{\mathscr{D}+\mathscr{B}}^{}(X)\,:=\, \!\!\bigoplus_{i\in \{1-d,\dots,0\}}~\bigoplus_{\und{a}\in\{0,1\}^{d+1}} 
\g^{i}\otimes \Big(\Omega^{\und{a}}(M)\,\widehat{\otimes}\,\Omega^{0,\bullet}\big(C,L_{D^i_{\und{a}}+B^i_{\und{a}}}\big)\Big)
\end{flalign}
by twisting the component-wise Dolbeault complexes not only with the holomorphic line
bundles from the singularity structure $\mathscr{D}$ but also with those
from the boundary condition $\mathscr{B}$. (Recall that $L_{D^i_{\und{a}}}\otimes L_{B^i_{\und{a}}}\cong
L_{D^i_{\und{a}}+B^i_{\und{a}}}$.) Note that $\E_{\mathscr{D}+\mathscr{B}}^{}(X)$ inherits the structure 
of a $\bbZ$-graded vector space and that the properties (i) and (ii) in Definition \ref{def:bdycondition}
are designed precisely in such a way that $\E_{\mathscr{D}+\mathscr{B}}^{}(X)$ 
can be endowed with the de Rham differential $\dd_M$ and the extension from Remark \ref{rem:tensorproduct} of the
$L_\infty$-structure $\ell^\g$. (Properties (iii) and (iv) in Definition \ref{def:bdycondition} are 
not relevant at the present point, but they will play an important role later when we discuss cyclic structures
and extract a $(d+1)$-dimensional integrable field theory.)
This leads us to the following definition.
\begin{defi}\label{def:CSbdy}
Let $(\g,\ell^\g)$ be a finite-dimensional $d$-term $L_\infty$-algebra
endowed with a cyclic structure $\pair{\cdot}{\cdot}^\g$ of degree $d-1$,
and let $\omega$ be a meromorphic $1$-form on $C$. Choose any 
singularity structure $\mathscr{D}$ as in Definition \ref{def:singularity}
and local boundary condition $\mathscr{B}$ as in Definition \ref{def:bdycondition}.
The $L_\infty$-algebra $\big(\E_{\mathscr{D}+\mathscr{B}}^{}(X),\ell\big)$ 
of \textit{topological-holomorphic $(\g,\ell^\g)$-Chern-Simons theory with singularities
$\mathscr{D}$ and boundary conditions $\mathscr{B}$} on the $(d+3)$-dimensional manifold $X=M\times C$ 
is defined by the $\bbZ$-graded vector space \eqref{eqn:bdytwisted} and 
the $L_\infty$-structure given by $\ell_1 := \dd_M + \delbar + \ell_1^\g$
and $\ell_n := \ell_n ^\g$, for all $n\geq 2$.
\end{defi}

Let us now construct a cyclic structure of degree $-3$ on the compactly supported sections
of the $L_\infty$-algebra $\big(\E_{\mathscr{D}+\mathscr{B}}^{}(X),\ell\big)$ from Definition \ref{def:CSbdy},
which in particular defines an action functional for this theory.
Let us start with observing that, as a consequence
of the properties listed in Definitions \ref{def:singularity} and \ref{def:bdycondition},
the cyclic structure $\pair{\cdot}{\cdot}^\g : \g^{\otimes 2}\to \bbC$
of degree $d-1$ on the structure $L_\infty$-algebra $(\g,\ell^\g)$ induces
via the construction in Proposition \ref{prop:cyclictensor} (recall also Remark \ref{rem:cyclictensor}) 
a degree $d-1$ linear map
\begin{flalign}\label{eqn:densitizedcyclic}
\xymatrix@C=3em{
\pair{\cdot}{\cdot}^\g\,:\,\E_{\mathscr{D}+\mathscr{B}}(X)^{\otimes 2} \ar[r]
~&~\Omega^{\bullet}(M)\,\widehat{\otimes}\,\Omega^{0,\bullet}\big(C,L_{(\omega)}\big)
}
\end{flalign}
satisfying the $L_\infty$-equivariance properties stated in this proposition.
The reason why this map takes values in the holomorphic line bundle $L_{(\omega)}\to C$
associated with the divisor of $\omega$ is as follows: Since 
$\pair{\cdot}{\cdot}^\g : \g^{\otimes 2}\to \bbC$
has degree $d-1$, it can only pair non-trivially among the 
complementary components $\g^i$ and $\g^{1-d-i}$, for all $i\in\{1-d,\dots,0\}$.
Given any $\und{a},\und{b}\in\{0,1\}^{d+1}$ with $\und{a}+ \und{b} \in 
\{0,1\}^{d+1}$, such that the $\wedge$-product between $\Omega^{\und{a}}(M)$
and $\Omega^{\und{b}}(M)$ does not necessarily vanish, we compute
\begin{flalign}
D^{i}_{\und{a}} + B^{i}_{\und{a}} + D^{1-d-i}_{\und{b}} + B^{1-d-i}_{\und{b}}\,\leq\,
D^{i}_{\und{a}} + B^{i}_{\und{a}} + D^{1-d-i}_{\und{1}-\und{a}} + B^{1-d-i}_{\und{1}-\und{a}}\,=\,
(\omega)_{0} + (\omega)_{\infty} \,=\,(\omega)\quad,
\end{flalign}
where the first step uses property (i) and the second step uses property (iii)
of Definitions \ref{def:singularity} and \ref{def:bdycondition}. 
These partial orderings of divisors induce holomorphic line bundle 
morphisms $L_{D^{i}_{\und{a}} + B^{i}_{\und{a}}}\otimes
L_{D^{1-d-i}_{\und{b}} + B^{1-d-i}_{\und{b}}} \cong 
L_{D^{i}_{\und{a}} + B^{i}_{\und{a}} + D^{1-d-i}_{\und{b}} + B^{1-d-i}_{\und{b}}}
\to L_{(\omega)}$, which allow us to define \eqref{eqn:densitizedcyclic}.
Post-composing this map with the $\wedge$-product of the given meromorphic
$1$-form $\omega$ defines a degree $d$ linear map
\begin{flalign}\label{eqn:densitizedcyclicomega}
\xymatrix@C=3em{
\omega\wedge \pair{\cdot}{\cdot}^\g\,:\,
\E_{\mathscr{D}+\mathscr{B}}(X)^{\otimes 2} \ar[r] ~&~\Omega^\bullet(X)
}
\end{flalign}
to the full de Rham algebra of $X = M\times C$.
\begin{propo}\label{prop:cyclicCS}
The linear map
\begin{flalign}\label{eqn:cyclicCS}
\xymatrix@C=2.5em{
\pair{\cdot}{\cdot}_{\omega}^{}\,:\,
\E_{\mathscr{D}+\mathscr{B},\cc}(X)^{\otimes 2}\ar[rr]^-{\omega\wedge \pair{\cdot}{\cdot}^\g}
~&&~\Omega_\cc^\bullet(X)\ar[r]^-{\int_X}~&~ \bbC
}
\end{flalign}
defined by restricting \eqref{eqn:densitizedcyclicomega} to compactly supported
sections and composing with the integration map on $X$ defines
a cyclic structure of degree $-3$ on the $L_\infty$-subalgebra 
$\big(\E_{\mathscr{D}+\mathscr{B},\cc}(X),\ell\big)\subseteq 
\big(\E_{\mathscr{D}+\mathscr{B}}(X),\ell\big)$ of compactly supported sections.
\end{propo}
\begin{proof}
Since $X = M\times C$ is $(d+3)$-dimensional, the integration map $\int_X$ has degree $-(d+3)$,
hence \eqref{eqn:cyclicCS} has degree $d-(d+3)=-3$. Non-degeneracy follows from the
observation that the pairing entering the top-degree form integrand 
in \eqref{eqn:cyclicCS} is point-wise non-degenerate
as a consequence of property (iii) of Definitions \ref{def:singularity} 
and \ref{def:bdycondition}. Invariance of $\pair{\cdot}{\cdot}_{\omega}^{}$ under
the $L_\infty$-brackets $\ell_n$, for $n\geq 2$, follows directly from Proposition \ref{prop:cyclictensor},
while for invariance under the $\ell_1$-bracket one further uses that 
$\omega\wedge (\dd_M + \delbar)(\,\cdot\,) = -\dd_X\big(\omega\wedge(\,\cdot\,)\big)$, since
$\omega$ is a meromorphic $1$-form on $C$, and Stokes' theorem $\int_X\dd_X(\,\cdot\,)=0$.
\end{proof}

\begin{rem}\label{rem:cyclicCS}
Let us describe more explicitly the action functional 
which is defined by the cyclic structure from Proposition \ref{prop:cyclicCS}
on the topological-holomorphic $(\g,\ell^\g)$-Chern-Simons theory with singularities $\mathscr{D}$ and boundary 
conditions $\mathscr{B}$ from Definition \ref{def:CSbdy}.
This action is defined on degree $1$ elements $\mathcal{A} \in \E_{\mathscr{D}+\mathscr{B},\cc}(X)^1$
with compact support and it reads as
\begin{flalign}\label{eqn:CSaction}
\nn S_\omega(\mathcal{A})\,&:=\,\sum_{n\geq 1}\tfrac{1}{(n+1)!}\,\pair{\mathcal{A}}{\ell_n\big(\mathcal{A}^{\otimes n}\big)}_\omega^{}\\
\,&=\,\int_X\omega\wedge \Big(\tfrac{1}{2}\,\pair{\mathcal{A}}{(\dd_M+\delbar)\mathcal{A}}^\g + \sum_{n\geq 1}\tfrac{1}{(n+1)!}\, 
\pair{\mathcal{A}}{\ell_n^\g\big(\mathcal{A}^{\otimes n}\big)}^\g\Big)\quad.
\end{flalign}
This action clearly resembles the $4$-dimensional action of Costello and Yamazaki \cite{CY3}, 
and more importantly it provides a generalization to $(d+3)$-dimensional manifolds $X=M\times C$.
It is worthwhile to emphasize that $\mathcal{A} \in \E_{\mathscr{D}+\mathscr{B},\cc}(X)^1$ is a short-hand notation
for a family of differential forms \eqref{eqn:bdytwisted}, with form degrees reaching from $1$ to $d$.
Hence, the action \eqref{eqn:CSaction} describes a rich theory in which
these components couple among each other with interactions determined by 
the brackets $\ell_n^\g$ of the structure $L_\infty$-algebra $(\g,\ell^\g)$. 
\end{rem}

To identify the integrable field theory on the $(d+1)$-dimensional manifold
$M$ which is associated with the $L_\infty$-algebra $\big(\E_{\mathscr{D}+\mathscr{B}}^{}(X),\ell\big)$, 
we use the homological perturbation and homotopy transfer techniques from \cite{BSV2}.
As a consequence of the properties of singularity structures and boundary conditions
in Definitions \ref{def:singularity} and \ref{def:bdycondition}, 
and in particular property (iv) of the latter, the component-wise
divisor-twisted Dolbeault cohomologies in \eqref{eqn:bdytwisted} fall into the
following three cases:
\begin{itemize}
\item[(1)] Suppose that $i\in\{1-d,\dots,0\}$ and $\und{a}\in \{0,1\}^{d+1}$ are
such that $i + \vert\und{a}\vert \leq 0$. Writing this inequality as $\vert\und{a}\vert \leq -i \leq d-1$,
we observe that $\mathsf{suc}(\und{a}) \neq \varnothing$ since $M$ is $(d+1)$-dimensional. 
Picking any $\und{b}\in \mathsf{suc}(\und{a})$ and using property (i) of Definitions 
\ref{def:singularity} and \ref{def:bdycondition}, we obtain
the estimate $\deg\big(D^i_{\und{a}}+B^i_{\und{a}}\big)\leq \deg\big(D^i_{\und{b}}+B^i_{\und{b}}\big)$.
Since this step increases the total degree $i + \vert\und{b}\vert  = i + \vert\und{a}\vert +1$ by $1$, 
we can repeat it (if necessary) until the total degree reaches $1$. 
It then follows from property (iv) of Definition \ref{def:bdycondition} that 
\begin{flalign}
\deg\big(D^i_{\und{a}}+B^i_{\und{a}}\big)\,\leq\,-1\quad.
\end{flalign}
The divisor-twisted Dolbeault cohomologies in this case are thus given by
\begin{flalign}
\mathsf{H}^0\Omega^{0,\bullet}\big(C,L_{D^i_{\und{a}}+B^i_{\und{a}}}\big)\,\cong \,0\quad,\qquad
\mathsf{H}^1\Omega^{0,\bullet}\big(C,L_{D^i_{\und{a}}+B^i_{\und{a}}}\big)\,\cong \,\bbC^{N_{\und{a}}^i}\quad,
\end{flalign}
with $N_{\und{a}}^i := -1-\deg\big(D^i_{\und{a}}+B^i_{\und{a}}\big)\geq 0$.

\item[(2)] Suppose that $i\in\{1-d,\dots,0\}$ and $\und{a}\in \{0,1\}^{d+1}$ are
such that $i + \vert\und{a}\vert =1$. It then follows directly 
from property (iv) of Definition \ref{def:bdycondition} that the
divisor-twisted Dolbeault cohomologies vanish in this case
\begin{flalign}
\mathsf{H}^0\Omega^{0,\bullet}\big(C,L_{D^i_{\und{a}}+B^i_{\und{a}}}\big)\,\cong \,0\quad,\qquad
\mathsf{H}^1\Omega^{0,\bullet}\big(C,L_{D^i_{\und{a}}+B^i_{\und{a}}}\big)\,\cong \,0\quad.
\end{flalign}

\item[(3)] Suppose that $i\in\{1-d,\dots,0\}$ and $\und{a}\in \{0,1\}^{d+1}$
are such that $i + \vert\und{a}\vert \geq 2$. Then the complementary
labels $1-d-i\in \{1-d,\dots,0\}$ and $\und{1}-\und{a}\in \{0,1\}^{d+1}$ satisfy
\begin{flalign}
1-d-i + \vert \und{1} - \und{a}\vert \,=\,2 - (i+ \vert\und{a}\vert) \,\leq 0\,\quad,
\end{flalign}
where we used that $\vert\und{1}\vert = d+1$ since $M$ is $(d+1)$-dimensional.
From the argument in (1) we obtain the estimate
$\deg\big(D^{1-d-i}_{\und{1}-\und{a}}+B^{1-d-i}_{\und{1}-\und{a}}\big)\leq -1$,
which together with property (iii) of Definitions 
\ref{def:singularity} and \ref{def:bdycondition} implies that
\begin{flalign}
\nn \deg\big(D^i_{\und{a}} + B^i_{\und{a}}\big) \,&=\, 
\deg(\omega) - \deg\big(D^{1-d-i}_{\und{1}-\und{a}}+B^{1-d-i}_{\und{1}-\und{a}}\big)\\
\,&=\, -2 - \deg\big(D^{1-d-i}_{\und{1}-\und{a}}+B^{1-d-i}_{\und{1}-\und{a}}\big)\geq -1\quad,
\end{flalign}
where we used that $\deg(\omega) = -2$ for any meromorphic $1$-form on $C=\CP$.
The divisor-twisted Dolbeault cohomologies in this case are thus given by
\begin{flalign}
\mathsf{H}^0\Omega^{0,\bullet}\big(C,L_{D^i_{\und{a}}+B^i_{\und{a}}}\big)\,\cong \,\bbC^{N^{i}_{\und{a}}}\quad,\qquad
\mathsf{H}^1\Omega^{0,\bullet}\big(C,L_{D^i_{\und{a}}+B^i_{\und{a}}}\big)\,\cong \,0\quad,
\end{flalign}
with $N_{\und{a}}^i := 1+ \deg\big(D^i_{\und{a}} + B^i_{\und{a}}\big) \geq 0$.
\end{itemize}
\begin{rem}\label{rem:dimensions}
Observe that, as a direct consequence of property (iii) of Definitions 
\ref{def:singularity} and \ref{def:bdycondition}, the dimensions $N_{\und{a}}^i  = N^{1-d-i}_{\und{1}-\und{a}}$
associated with any two complementary labels coincide. 
\end{rem}

Choosing now any family of continuous strong deformation retracts (e.g.\ from Hodge theory)
\begin{equation}\label{eqn:defretbdy0}
\begin{tikzcd}
\big(\mathsf{H}^\bullet\Omega^{0,\bullet}\big(C,L_{D^i_{\und{a}}+B^i_{\und{a}}}\big) ,0\big)
\ar[rr,shift right=-1ex,"i^i_{\und{a}}"] && \ar[ll,shift right=-1ex,"p^i_{\und{a}}"] 
\big(\Omega^{0,\bullet}(C,L_{D^i_{\und{a}}+B^i_{\und{a}}}),\delbar \big) \ar[loop,out=-17,in=15,distance=30,swap,"h^i_{\und{a}}"]
\end{tikzcd}\quad,
\end{equation} 
for all $i\in\{1-d,\dots,0\}$ and $\und{a}\in\{0,1\}^{d+1}$, we obtain via the 
totalization construction in \cite[Appendix A.1]{BSV2} a strong deformation retract
\begin{equation}\label{eqn:defretbdy}
\begin{tikzcd}[column sep =large]
\big(\F(M),\ell_1^\prime\big) \ar[rr,shift right=-1ex,"\widetilde{i}\,:=\,i+h(\dd_M+\ell_1^\g)i"] && \ar[ll,shift right=-1ex,"\widetilde{p}\,:=\,p + p(\dd_M+\ell_1^\g)h"] 
\big(\E_{\mathscr{D}+\mathscr{B}}^{}(X),\ell_1 \big) \ar[loop,out=-23,in=23,distance=30,swap,"h"]
\end{tikzcd}
\end{equation} 
from the underlying cochain complex of the $L_\infty$-algebra from Definition \ref{def:CSbdy}
to the cochain complex defined by 
\begin{subequations}\label{eqn:bdyM}
\begin{flalign}
\F(M)\,=\,
\bigoplus_{i+\vert\und{a}\vert\leq 0}
\g^{i}\otimes \Omega^{\und{a}}(M) \otimes\bbC^{N_{\und{a}}^i}[-1]
\oplus  
\bigoplus_{i+\vert\und{a}\vert\geq 2}
\g^{i}\otimes \Omega^{\und{a}}(M) \otimes\bbC^{N_{\und{a}}^i}
\quad,
\end{flalign}
endowed with the $\bbZ$-graded vector space structure inherited from $\g$ and $\Omega^\bullet(M)$, 
and differential 
\begin{flalign}
\ell_1^\prime \,:=\, p\big(\dd_M +\ell_1^\g\big)i + p\big(\dd_M +\ell_1^\g\big)h\big(\dd_M +\ell_1^\g\big)i\quad.
\end{flalign}
\end{subequations}
Note that the description of the graded vector space 
$\F(M)$ uses our cohomology computations from the itemization above.
The notation $\bbC^{N_{\und{a}}^i}[-1]$ means that the vector
space $\bbC^{N_{\und{a}}^i}$ gets assigned cohomological degree $+1$, 
which is due to the fact that it arises as the first Dolbeault cohomology.
The following result is then a direct consequence of the homotopy
transfer theorem for $L_\infty$-algebras, see e.g.\ \cite[Chapter 10.3]{LodayVallette}.
\begin{propo}\label{prop:IFT}
The $L_\infty$-algebra 
$\big(\E_{\mathscr{D}+\mathscr{B}}^{}(X),\ell\big)$ from Definition \ref{def:CSbdy}
admits a weakly equivalent description in terms of the $L_\infty$-algebra
$\big(\F(M),\ell^\prime\big)$ whose underlying cochain complex is given
by \eqref{eqn:bdyM} and whose $L_\infty$-brackets $\ell_n^\prime$,
for $n\geq 2$, are determined by homotopy transfer along the strong deformation retract 
\eqref{eqn:defretbdy}. This weak equivalence is implemented by an $\infty$-quasi-isomorphism
\begin{flalign}
\xymatrix{
\widetilde{i}_\infty\,:\,\big(\F(M),\ell^\prime\big) \ar@{~>}[r]^-{\sim}~&~\big(\E_{\mathscr{D}+\mathscr{B}}^{}(X),\ell\big)
}
\end{flalign}
extending the quasi-isomorphism $\widetilde{i}=i+h(\dd_M+\ell_1^\g)i$ of cochain complexes from \eqref{eqn:defretbdy}.
\end{propo}

\begin{rem}\label{rem:IFT}
In order to understand better the kind of field theory on the $(d+1)$-dimensional
manifold $M$ which is described by the weakly equivalent $L_\infty$-algebra $\big(\F(M),\ell^\prime\big)$
from Proposition \ref{prop:IFT}, let us unravel its key features. From \eqref{eqn:bdyM}, we observe
that a degree $1$ element $\Phi\in \F(M)^1$ is concretely given by a family of forms
\begin{flalign}\label{eqn:IFTformcomponents}
\Phi_{\und{a}} \,\in\, \g^{-\vert\und{a}\vert}\otimes \Omega^{\und{a}}(M)\otimes \bbC^{N^{-\vert\und{a}\vert}_{\und{a}}}\quad,
\end{flalign}
for all $\und{a}\in\{0,1\}^{d+1}$ such that $\vert \und{a}\vert \leq d-1$, which take
values in the tensor product of the component $\g^{-\vert\und{a}\vert}$ of the 
structure $L_\infty$-algebra $(\g,\ell^\g)$ and a vector space of dimension 
$N^{-\vert\und{a}\vert}_{\und{a}}\geq 0$. These are the fields
of the theory described by the $L_\infty$-algebra $\big(\F(M),\ell^\prime\big)$. Their dynamics
is given by the Maurer-Cartan equation, which in the present case reads as follows
\begin{flalign}\label{eqn:MCequationIFT}
p\,\dd_M \,h\,\dd_M\,i\big(\Phi\big) + p\,\big(\dd_M \,h \,\ell_1^\g + \ell_1^\g\,h\,\dd_M\big) \,i\big(\Phi\big)
+ p\,\ell_1^\g \,h\,\ell_1^\g \,i\big(\Phi\big) + \sum_{n\geq 2}\tfrac{1}{n!}\,\ell_n^\prime\big(\Phi^{\otimes n}\big)\,=\,0\quad,
\end{flalign}
where we used that, as a consequence of the vanishing cohomologies in item (2) from above, 
the term $p\big(\dd_M +\ell_1^\g\big)i$ in the differential $\ell_1^\prime$ from \eqref{eqn:bdyM}
acts trivially on such $\Phi$. We have ordered the terms in \eqref{eqn:MCequationIFT} such that
they display clearly a second-order, first-order and zeroth-order linear differential operator on $M$
and an interaction term which is determined by the $L_\infty$-brackets $\ell_n^\prime$, for $n\geq 2$.
Hence, the fields described by the $L_\infty$-algebra $\big(\F(M),\ell^\prime\big)$ satisfy a non-linear 
second-order partial differential equation that mixes the various form components in \eqref{eqn:IFTformcomponents}.
\sk

In the case of $d\geq 2$, this field theory on the $(d+1)$-dimensional manifold
$M$ comes with $k$-gauge transformations $\epsilon_{(1-k)}\in \F(M)^{1-k}$,
for all $k=1,\dots,d-1$ reaching up to $d-1$, while for the traditional case
of $d=1$ there are no non-trivial gauge transformations. These (higher)
gauge transformations can be described easily by specializing 
the general presentation from Remark \ref{rem:Linftyalgebra}, 
but we do not find this instructive so we will not spell out the details.
Note that the categorical depth of the gauge transformations in the 
$L_\infty$-algebra $\big(\F(M),\ell^\prime\big)$ is reduced by $1$
with respect to the original topological-holomorphic $(\g,\ell^\g)$-Chern-Simons theory
$\big(\E(X),\ell\big)$ on $X = M\times C$ that we have started from in Section \ref{sec:higherCS},
which has $k$-gauge transformations reaching up to $k=d$.
\end{rem}

We announced above that the $L_\infty$-algebra $\big(\F(M),\ell^\prime\big)$
describes an integrable field theory on the $(d+1)$-dimensional manifold $M$, 
so let us now substantiate this claim by showing how one can assign to it 
higher Lax connections as in Section \ref{sec:higherLax}. Since the divisors
$\mathscr{B} = \{B^{i}_{\und{a}} \leq 0\}$ entering our Definition \ref{def:bdycondition}
of local boundary conditions are all non-positive, we have partial divisor orderings
$D^i_{\und{a}} + B^{i}_{\und{a}}\leq D^i_{\und{a}}$ and their associated holomorphic
line bundle morphisms $L_{D^i_{\und{a}} + B^{i}_{\und{a}}}\to L_{D^i_{\und{a}}}$, for
all $i\in\{1-d,\dots,0\}$ and $\und{a}\in\{0,1\}^{d+1}$. These define a (strict) morphism 
\begin{flalign}
\xymatrix{
\big(\E_{\mathscr{D}+\mathscr{B}}(X),\ell\big)\ar[r]~&~\big(\E_{\mathscr{D}}(X),\ell\big)
}
\end{flalign}
from the $L_\infty$-algebra $\big(\E_{\mathscr{D}+\mathscr{B}}(X),\ell\big)$
of topological-holomorphic $(\g,\ell^\g)$-Chern-Simons theory with singularities
$\mathscr{D}$ and boundary conditions $\mathscr{B}$ from Definition \ref{def:CSbdy} 
to the $L_\infty$-algebra $\big(\E_{\mathscr{D}}(X),\ell\big)$
of topological-holomorphic $(\g,\ell^\g)$-Chern-Simons theory with singularities
$\mathscr{D}$ but no boundary conditions from Definition \ref{def:CSsingularity}. 
Combining this with Propositions \ref{prop:higherLax} 
and \ref{prop:IFT} leads immediately to the following result
which gives a precise statement about the integrability of $\big(\F(M),\ell^\prime\big)$.
\begin{theo}\label{theo:mapLax}
Our constructions above define a (weak) $L_\infty$-morphism
\begin{flalign}
\begin{gathered}
\xymatrix@C=4em{
\ar@{~>}[d]_-{\widetilde{i}_{\infty}}^-{\sim} \big(\F(M),\ell^\prime\big) \ar@{~>}[r]^-{}~&~\big(\L(M),\ell^\prime\big) \\
\big(\E_{\mathscr{D}+\mathscr{B}}(X),\ell\big) \ar[r] ~&~\big(\E_{\mathscr{D}}(X),\ell\big)\ar@{~>}[u]^-{\sim}_-{i^{-1}}
}
\end{gathered}
\end{flalign}
mapping from the $L_\infty$-algebra $\big(\F(M),\ell^\prime\big)$ 
describing a field theory on the $(d+1)$-dimensional manifold $M$  
(see Proposition \ref{prop:IFT} and Remark \ref{rem:IFT}) to the 
$L_\infty$-algebra $\big(\L(M),\ell^\prime\big)$ describing
higher Lax connections on $M$ (see Proposition \ref{prop:higherLax} and Remark \ref{rem:higherLax}).
\end{theo}
\begin{rem}\label{rem:mapLax}
It is precisely this $L_\infty$-morphism which endows $\big(\F(M),\ell^\prime\big)$ with the structure
of an integrable field theory on the $(d+1)$-dimensional manifold $M$, 
in the sense that it assigns to every field $\Phi$, i.e.\ a Maurer-Cartan element 
of the $L_\infty$-algebra $\big(\F(M),\ell^\prime\big)$, 
a higher Lax connection $A$, i.e.\ a Maurer-Cartan element of the $L_\infty$-algebra
$\big(\L(M),\ell^\prime\big)$, see also Remarks \ref{rem:higherLax} and \ref{rem:IFT}.
\end{rem}

We conclude this section by showing that 
the cyclic structure $\pair{\cdot}{\cdot}_\omega^{}$ of degree $-3$ 
on the $L_\infty$-subalgebra $\big(\E_{\mathscr{D}+\mathscr{B},\cc}(X),\ell\big)\subseteq 
\big(\E_{\mathscr{D}+\mathscr{B}}(X),\ell\big)$ from Proposition \ref{prop:cyclicCS}
admits a transfer to a cyclic structure of degree $-3$ on the $L_\infty$-subalgebra
$\big(\F_\cc(M),\ell^\prime\big)\subseteq \big(\F(M),\ell^\prime\big)$ of compactly supported
sections of the $L_\infty$-algebra from Proposition \ref{prop:IFT}.
This provides in particular an action functional for the latter 
$(d+1)$-dimensional integrable field theory.
\begin{propo}\label{prop:cyclictransfer}
There exists a choice for the component-wise continuous strong deformation retracts in
\eqref{eqn:defretbdy0} such that the totalized strong deformation retract \eqref{eqn:defretbdy} 
satisfies the following compatibility conditions:
\begin{itemize}
\begin{subequations}\label{eqn:cycliccompatible}
\item[(i)] For all $\alpha\in\mathrm{im}(\widetilde{i})\subseteq \E_{\mathscr{D}+\mathscr{B},\cc}(X)$ 
and $\beta\in \mathrm{ker}(\widetilde{p})\subseteq \E_{\mathscr{D}+\mathscr{B},\cc}(X)$,
\begin{flalign}\label{eqn:cycliccompatible1}
\pair{\alpha}{\beta}_{\omega}^{}\,=\,0\quad.
\end{flalign}

\item[(ii)] For all $\alpha,\beta\in \E_{\mathscr{D}+\mathscr{B},\cc}(X)$,
\begin{flalign}\label{eqn:cycliccompatible2}
\pair{h(\alpha)}{\beta}_{\omega}^{} \,=\,(-1)^{\vert\alpha\vert} \,\pair{\alpha}{h(\beta)}_{\omega}^{}\quad.
\end{flalign}
\end{subequations}
\end{itemize}
By the results of \cite{LazarevHodge,LazarevCS}, this implies that the transferred pairing
\begin{flalign}
\xymatrix{
\pair{\cdot}{\cdot}_{\omega}^\prime\,:=\,\pair{\cdot}{\cdot}_\omega^{}\circ (\widetilde{i}\otimes\widetilde{i})\,:\, 
\F_\cc(M)\otimes \F_\cc(M)\ar[r]~&~\bbC
}
\end{flalign}
defines a cyclic structure of degree $-3$ on the $L_\infty$-subalgebra 
$\big(\F_\cc(M),\ell^\prime\big)\subseteq \big(\F(M),\ell^\prime\big)$ of compactly supported
sections of the $L_\infty$-algebra from Proposition \ref{prop:IFT}.
\end{propo}
\begin{proof}
By the construction in \cite[Appendix A.3]{BSV2}, one can choose
the component-wise continuous strong deformation retracts in
\eqref{eqn:defretbdy0} such that their \textit{undeformed} totalization
\begin{equation}\label{def:undefSDR}
\begin{tikzcd}
\big(\F(M),0\big) \ar[rr,shift right=-1ex,"i"] && \ar[ll,shift right=-1ex,"p"] 
\big(\E_{\mathscr{D}+\mathscr{B}}^{}(X),\delbar\big) \ar[loop,out=-25,in=25,distance=30,swap,"h"]
\end{tikzcd}
\end{equation} 
satisfies the compatibility conditions \eqref{eqn:cycliccompatible}.
We now show that its deformation $(\widetilde{i},\widetilde{p},h)$
given in \eqref{eqn:defretbdy} satisfies these compatibility conditions too.
While the second property \eqref{eqn:cycliccompatible2} holds true automatically
since the cochain homotopy $h$ does not receive any deformation, verifying
the first property \eqref{eqn:cycliccompatible1} requires a brief argument.
Given any $\alpha\in\mathrm{im}(\widetilde{i})\subseteq \E_{\mathscr{D}+\mathscr{B},\cc}(X)$ 
and $\beta\in \mathrm{ker}(\widetilde{p})\subseteq \E_{\mathscr{D}+\mathscr{B},\cc}(X)$,
we can write $\alpha = \widetilde{i}(\rho) = i(\rho) + h(\dd_M + \ell_1^\g)i(\rho)$ for some 
$\rho\in \F_\cc(M)$ and compute
\begin{flalign}
\nn \pair{\alpha}{\beta}_{\omega}^{} \,&=\, 
\pair{i(\rho)}{\beta}_{\omega}^{} + \pair{h(\dd_M + \ell_1^\g)i(\rho)}{\beta}_{\omega}^{}\\
\nn \,&=\,\pair{i(\rho)}{\beta}_{\omega}^{} + (-1)^{\vert \rho\vert+1}\pair{(\dd_M + \ell_1^\g)i(\rho)}{h(\beta)}_{\omega}^{}\\
\nn \,&=\, \pair{i(\rho)}{\beta}_{\omega}^{} - (-1)^{\vert \rho\vert+1}\, (-1)^{\vert \rho\vert}\,
\pair{i(\rho)}{(\dd_M + \ell_1^\g)h(\beta)}_{\omega}^{}\\
\,&=\,\pair{i(\rho)}{\beta + (\dd_M + \ell_1^\g)h(\beta)}_{\omega}^{} \,=\,0\quad,
\end{flalign}
where in the second step we used \eqref{eqn:cycliccompatible2} and 
in the third step we used that $\pair{\cdot}{\cdot}_{\omega}^{}$ 
is, by construction, not only compatible with the total differential $\delbar + \dd_M + \ell_1^\g$,
but also with its individual summands $\delbar$, $\dd_M$ and $\ell_1^\g$.
In the last step we used that $p\big(\beta + (\dd_M + \ell_1^\g)h(\beta)\big) = \widetilde{p}(\beta)=0$,
by hypothesis, hence the displayed term in the last line vanishes since the undeformed totalization \eqref{def:undefSDR}
has been previously chosen so as to satisfy the compatibility condition \eqref{eqn:cycliccompatible1}.
\end{proof}

%%%%%%%%%%%%%%%%%%%%%%%%%%%%%%%%%%%%%%%%%%%%%%%%
%%%%%%%%%%%%%%%%%%%%%%%%%%%%%%%%%%%%%%%%%%%%%%%%

\section{\label{sec:examples}Examples}
\subsection{\label{subsec:exsglbdy}Singularity structures and local boundary conditions}
In this subsection we present explicit examples of singularity structures and local boundary conditions
in the sense of Definitions \ref{def:singularity} and \ref{def:bdycondition}, for any choice of
integer $d\in \bbZ_{\geq 1}$. Our choice of examples is motivated by their simplicity and we do not claim
that they exhaust all possibilities. 
\begin{ex}\label{ex:example1}
Consider a meromorphic $1$-form $\omega$ on $C=\CP$ which has $2d$ zeros,
taking the form of $d+1$ simple zeros $q_1,\dots,q_{d+1}\in C$ and a zero $q_0\in C$
of order $d-1$, and consequently $2(d+1)$ poles $p_1,\dots,p_{d+1},\tilde{p}_{1},\dots,\tilde{p}_{d+1}\in C$.
Setting
\begin{flalign}
D_{\und{a}}^i\,:=\,\begin{cases}
0 ~&~,~~\text{for }\vert\und{a}\vert = 0\\
\big(\vert \und{a}\vert -1\big)\,q_0 + \sum_{k=1}^{d+1} a_k\,q_k  ~&~,~~\text{for } 1\leq \vert\und{a}\vert\leq d\\
(\omega)_0~&~,~~\text{for }\vert\und{a}\vert= d+1
\end{cases}
\quad,
\end{flalign}
for all $i\in\{1-d,\dots,0\}$ and $\und{a}\in\{0,1\}^{d+1}$, defines a 
singularity structure $\mathscr{D}$ in the sense of Definition \ref{def:singularity}.
Property (i) is satisfied by design. For property (ii), the non-trivial check is to consider
the case of $1\leq \vert\und{a}_k\vert\leq d$, for all $k=1,\dots,n$, where we have
\begin{flalign}
\nn \sum_{k=1}^n D_{\und{a}_k}^{i_k}\,&=\,\sum_{k=1}^n\bigg(\big(\vert \und{a}_k\vert-1\big)\,q_0 + \sum_{j=1}^{d+1} a_{kj}\,q_j \bigg)
\,=\,\bigg(\Big\vert \sum_{k=1}^n\und{a}_k\Big\vert-n\bigg)\,q_0 + \sum_{j=1}^{d+1}\sum_{k=1}^n a_{kj}\,q_j\\
\,&\leq\,\bigg(\Big\vert \sum_{k=1}^n\und{a}_k\bigg\vert-1\Big)\,q_0 + \sum_{j=1}^{d+1}\sum_{k=1}^n a_{kj}\,q_j\,=\,
D^{2-n+\sum_{k=1}^ni_k}_{\sum_{k=1}^n\und{a}_k}\quad.
\end{flalign}
For property (iii), the non-trivial check is again to consider the case of $1\leq \vert \und{a}\vert\leq d$, where we have
\begin{flalign}
D^i_{\und{a}} + D^{1-d-i}_{\und{1}-\und{a}}\,=\,\Big(\vert\und{a}\vert -1 + \vert \und{1}-\und{a}\vert -1\Big) \,q_0
+ \sum_{k=1}^{d+1} q_k \,=\,(d-1)\,q_0 + \sum_{k=1}^{d+1} q_k \,=\,(\omega)_0\quad.
\end{flalign}
Setting further
\begin{flalign}
B_{\und{a}}^i\,:=\,\begin{cases}
(\omega)_\infty~&~,~~\text{for }i+\vert\und{a}\vert \leq 0\\
-\sum_{k=1}^{d+1} a_k\,\big(p_k+\tilde{p}_k\big)~&~,~~\text{for }i+\vert\und{a}\vert =1\\
0 ~&~,~~\text{for }i+\vert\und{a}\vert \geq 2
\end{cases}\quad,
\end{flalign}
for all $i\in\{1-d,\dots,0\}$ and $\und{a}\in\{0,1\}^{d+1}$, defines a corresponding
local boundary condition in the sense of Definition \ref{def:bdycondition}. Properties (i) and (iii)
are satisfied by design, and property (iv) is shown by the following calculation: For 
total degree $i+\vert\und{a}\vert=1$, one necessarily has that 
$1\leq \vert\und{a}\vert \leq d$ since $i\in \{1-d,\dots,0\}$, hence
\begin{flalign}
\deg\big(D^i_{\und{a}} + B^i_{\und{a}}\big)\,=\, \vert \und{a}\vert -1 + \vert \und{a}\vert - 2\,\vert \und{a}\vert \,=\,-1\quad.
\end{flalign}
It remains to verify property (ii). Since $(\omega)_\infty\leq B^i_{\und{a}}\leq 0$, this property is clearly satisfied
whenever at least one of the total degrees $i_k +\vert \und{a}_k\vert\leq 0$ is non-positive, since in this case
the sum on the left-hand side of \eqref{eqn:bdypropii} can be estimated by $\leq (\omega)_\infty$. This 
reduces the problem to the case where $i_k +\vert \und{a}_k\vert\geq 1$, for all $k=1,\dots,n$.
In this case the total degree on the right-hand side of \eqref{eqn:bdypropii} is given by 
\begin{flalign}
2-n + \sum_{k=1}^{n} \big(i_k +\vert \und{a}_k\vert \big) \,\geq 2\quad,
\end{flalign}
which implies that the right-hand side of \eqref{eqn:bdypropii} is $0$ and hence the estimate holds true.
\end{ex}

\begin{ex}\label{ex:example2}
Consider a meromorphic $1$-form $\omega$ on $C=\CP$ which has $d+1$ simple zeros $q_1,\dots,q_{d+1}\in C$
and consequently $d+3$ poles, taking the form of $d+1$ simple poles
$p_1,\dots,p_{d+1}\in C$ and a double pole $p_0\in C$. Setting
\begin{flalign}
D_{\und{a}}^i\,:=\, \sum_{k=1}^{d+1} a_k\,q_k\quad,
\end{flalign}
for all $i\in\{1-d,\dots,0\}$ and $\und{a}\in\{0,1\}^{d+1}$, clearly defines a 
singularity structure $\mathscr{D}$ in the sense of Definition \ref{def:singularity}.
Setting further
\begin{flalign}
B_{\und{a}}^i\,:=\,\begin{cases}
(\omega)_\infty~&~,~~\text{for }i+\vert\und{a}\vert \leq 0\\
-p_0-\sum_{k=1}^{d+1}a_k\,p_k~&~,~~\text{for }i+\vert\und{a}\vert =1\\
0 ~&~,~~\text{for }i+\vert\und{a}\vert \geq 2
\end{cases}\quad,
\end{flalign}
for all $i\in\{1-d,\dots,0\}$ and $\und{a}\in\{0,1\}^{d+1}$, defines a corresponding
local boundary condition in the sense of Definition \ref{def:bdycondition}. Properties
(i), (iii) and (iv) are satisfied by design and property (ii) is shown by the same argument as
in Example \ref{ex:example1}.
\end{ex}

\begin{ex}\label{ex:example3}
In this last example we assume that the $d$-term structure $L_\infty$-algebra $(\g,\ell^\g)$ has 
trivial higher Lie brackets, i.e.\ $\ell^\g_n =0$ for all $n\geq 3$.
Consider a meromorphic $1$-form $\omega$ on $C=\CP$ which has $d+1$ simple zeros $q_1,\dots,q_{d+1}\in C$
and consequently $d+3$ poles, taking the form of two poles $p_1,p_2\in C$ of order $2$ 
and a pole $p_0\in C$ of order $d-1$. Setting
\begin{flalign}
D_{\und{a}}^i\,:=\, \sum_{k=1}^{d+1} a_k\,q_k\quad,
\end{flalign}
for all $i\in\{1-d,\dots,0\}$ and $\und{a}\in\{0,1\}^{d+1}$, defines the same singularity structure
as in Example \ref{ex:example2}. Setting further
\begin{flalign}
B_{\und{a}}^i\,:=\, -p_1 -p_2 + i\, p_0
\end{flalign}
for all $i\in\{1-d,\dots,0\}$ and $\und{a}\in\{0,1\}^{d+1}$, defines a corresponding
local boundary condition in the sense of Definition \ref{def:bdycondition}.
Properties (i), (iii) and (iv) are satisfied by design. As a consequence of our assumption
that $\ell^\g_n =0$, for all $n\geq 3$, property (ii) reduces to two simple 
checks for the two non-vanishing $L_\infty$-brackets $\ell_1^\g$ and $\ell_2^\g$.
The inequality $B^i_{\und{a}} \leq B^{i+1}_{\und{a}}$ for arity $1$ clearly holds true,
while the inequality for arity $2$ is shown by the calculation
\begin{flalign}
B^{i_1}_{\und{a}_1}+ B^{i_2}_{\und{a}_2}\,=\, -2p_1-2p_2+(i_1+i_2)\,p_0\,\leq\,
-p_1-p_2+(i_1+i_2)\,p_0\,=\,B^{i_1+i_2}_{\und{a}_1+\und{a}_2}\quad.
\end{flalign}
With a similar calculation one can understand why higher $L_\infty$-brackets $\ell^\g_n$,
for $n\geq 3$, are not admissible in this example: For the left-hand side of \eqref{eqn:bdypropii}
one finds that 
\begin{subequations}
\begin{flalign}
\sum_{k=1}^n B^{i_k}_{\und{a}_k} \,=\, -np_1-np_2 + \bigg(\sum_{k=1}^n i_k\bigg)\,p_0\quad,
\end{flalign}
while the right-hand side reads as 
\begin{flalign}
B^{2-n+\sum_{k=1}^n i_k}_{\sum_{k=1}^n\und{a}_k} \,=\,-p_1-p_2 + \bigg(2-n+\sum_{k=1}^n i_k\bigg)\,p_0\quad.
\end{flalign}
\end{subequations}
For $n\geq 3$, the prefactors of $p_0$ satisfy 
$\sum_{k=1}^n i_k > 2-n+\sum_{k=1}^n i_k$, which implies
that property (ii) is violated $\sum_{k=1}^n B^{i_k}_{\und{a}_k} \not\leq B^{2-n+\sum_{k=1}^n i_k}_{\sum_{k=1}^n\und{a}_k} $ 
in this case.
\end{ex}

\subsection{\label{subsec:exdynamics}Linearized dynamics of \texorpdfstring{$3$}{3}-dimensional integrable field theories}
In this subsection we specialize to the case of $d=2$, i.e.\ $(d+3=5)$-dimensional
topological-holomorphic Chern-Simons theory and its associated $(d+1=3)$-dimensional
integrable field theories from Proposition \ref{prop:IFT}. For the structure $(d=2)$-term $L_\infty$-algebra
$(\g,\ell^\g)$ we choose the shifted tangent $L_\infty$-algebra
$\g = T[1]\h$ of an ordinary Lie algebra $\h$, which is given explicitly
by $\g^0 = \g^{-1} =\h$ with $\ell^\g_n = 0$ trivial, for all $n\neq 2$,
and $\ell_2^\g=[\cdot,\cdot]_\h$ the Lie bracket of $\h$, acting on both the degree $0$ and $-1$
components of $\g$. A cyclic structure of degree $d-1=1$ on $(\g,\ell^\g)$
is then given by a choice of ad-invariant symmetric non-degenerate
pairing $\h\otimes\h\to \bbC$ on the Lie algebra $\h$. 
(Note that this agrees with the choices made in the earlier work 
\cite[Section 5.3]{SV} on $5d$ topological-holomorphic $2$-Chern-Simons theories.)
\sk

We will now provide a qualitative analysis of 
the underlying cochain complexes \eqref{eqn:bdyM}
of the $3$-dimensional integrable field theories which are associated
with our singularity structures $\mathscr{D}$ and local boundary conditions 
$\mathscr{B}$ from Subsection \ref{subsec:exsglbdy}. The form of these 
cochain complexes is dictated by our cohomology computations from the itemization before 
Remark \ref{rem:dimensions}, so one can already obtain valuable insights by 
simply computing the dimensions $N^i_{\und{a}}$ from the degrees
of the divisors entering $\mathscr{D}$ and $\mathscr{B}$.
\begin{ex}\label{ex:complex1}
The underlying cochain complex \eqref{eqn:bdyM} of the $3$-dimensional
integrable field theory which is associated with the singularity structure
and local boundary condition from Example \ref{ex:example1} in the case of $d=2$
takes the form
\begin{flalign}
\FFF(M)\,=\,
\left(\begin{gathered}
\xymatrix@R=0.2em{
0 \ar[r]^-{0}~&~ \Omega^0(M,\h^{5}) \ar[r]^-{\ell_1^\prime}~&~ \Omega^2(M,\h^{4}) \ar[r]^-{\ell_1^\prime}
~&~ \Omega^3(M,\h^{5})\\
\oplus ~&~\oplus ~&~\oplus ~&~\oplus \\
\Omega^0(M,\h^{5}) \ar[r]^-{\ell_1^\prime}~&~ \Omega^1(M,\h^{4}) \ar[r]^-{\ell_1^\prime}~&~ \Omega^3(M,\h^{5}) 
\ar[r]^-{0}~&~ 0
}
\end{gathered}
\right)\quad,
\end{flalign}
where the columns run from cohomological degree $0$ to $3$. The reason why the
two rows do not mix lies in our particular choice of structure $L_\infty$-algebra $\g = T[1]\h$
for which $\ell_1^\g =0$. 
\sk

The top row can be interpreted as a theory describing an $\h^{5}$-valued
scalar field $\Phi\in \Omega^0(M,\h^{5})$ with dynamics given by a second-order partial differential operator
$\ell^\prime_1(\Phi) \in \Omega^2(M,\h^{4})$ taking values in $\h^{4}$-valued $2$-forms.
The right-most arrow of the top-row describes first-order differential relations $\ell^\prime_1(\lambda^{+})=0$
for $\lambda^{+} \in \Omega^2(M,\h^{4})$, taking values in $\h^{5}$-valued $3$-forms,
that this second-order partial differential operator must satisfy.
With a simple degree counting argument, taking into account the independent differential form components, 
one finds that this system of differential equations is over-determined.
Indeed, there are $3\times \dim(\h)\times 4 = 12\times\dim(\h)$ equations with $1\times \dim(\h)\times 5 = 5\times \dim(\h)$
relations, which gives a total of $7\times \dim(\h)$ equations for the only $1\times \dim(\h)\times 5= 5\times \dim(\h)$
components of the field $\Phi$.
\sk

The bottom row can be interpreted as a theory describing an $\h^{4}$-valued $1$-form
gauge field $\lambda\in \Omega^{1}(M,\h^{4})$ with gauge transformations 
$\lambda \mapsto \lambda + \ell_1^\prime (\chi) $ parametrized by $\h^{5}$-valued
$0$-forms $\chi \in \Omega^0(M,\h^{5})$ and dynamics given by 
a second-order partial differential operator
$\ell^\prime_1(\lambda) \in \Omega^3(M,\h^{5})$ taking values in $\h^{5}$-valued $3$-forms.
Using again a simple degree counting argument as above, one finds that this system of 
differential equations is under-determined. Indeed, there are $3\times\dim(\h)\times 4 = 12\times\dim(\h)$
gauge field components with $1\times \dim(\h)\times 5 = 5\times \dim(\h)$ gauge symmetries, but there are
only $1\times \dim(\h)\times 5= 5\times \dim(\h)$ equations for the resulting $7\times \dim(\h)$ non-gauge redundant
components of the field $\lambda$.
\end{ex}

\begin{ex}\label{ex:complex2}
The underlying cochain complex \eqref{eqn:bdyM} of the $3$-dimensional
integrable field theory which is associated with the singularity structure
and local boundary condition from Example \ref{ex:example2} in the case of $d=2$
takes the form
\begin{flalign}
\FFF(M)\,=\,
\left(\begin{gathered}
\xymatrix@R=0.2em{
0 \ar[r]^-{0}~&~ \Omega^0(M,\h^{4}) \ar[r]^-{\ell_1^\prime}~&~ \Omega^2(M,\h^{3}) \ar[r]^-{\ell_1^\prime}
~&~ \Omega^3(M,\h^{4})\\
\oplus ~&~\oplus ~&~\oplus ~&~\oplus \\
\Omega^0(M,\h^{4}) \ar[r]^-{\ell_1^\prime}~&~ \Omega^1(M,\h^{3}) \ar[r]^-{\ell_1^\prime}~&~ \Omega^3(M,\h^{4}) 
\ar[r]^-{0}~&~ 0
}
\end{gathered}
\right)\quad,
\end{flalign}
where the columns run from cohomological degree $0$ to $3$. This cochain complex
takes the same shape as the one from Example \ref{ex:complex1}, so we can apply 
the same interpretations to this theory. The top row describes
an over-determined system consisting effectively of $3\times \dim(\h)\times 3-1\times \dim(\h)\times 4=5\times\dim(\h)$
differential equations for the only $1\times\dim(\h)\times 4 = 4\times \dim(\h)$ components of the scalar
field $\Phi \in \Omega^0(M,\h^{4}) $. Similarly, the bottom row describes
an under-determined system of only $1\times\dim(\h)\times 4=4\times\dim(\h)$ differential equations
for the $3\times\dim(\h)\times 3 - 1\times\dim(\h)\times 4 = 5\times\dim(\h)$ non-gauge redundant
components of the gauge field $\lambda\in \Omega^1(M,\h^{3})$.
\end{ex}

\begin{ex}\label{ex:complex3}
The underlying cochain complex \eqref{eqn:bdyM} of the $3$-dimensional
integrable field theory which is associated with the singularity structure
and local boundary condition from Example \ref{ex:example3} in the case of $d=2$
takes the form
\begin{flalign}\label{eqn:complex3}
\FFF(M)\,=\,
\left(\begin{gathered}
\xymatrix@R=0.2em{
0 \ar[r]^-{0}~&~ \Omega^0(M,\h) \ar[r]^-{\ell_1^\prime}~&~ \Omega^2(M,\h) \ar[r]^-{\ell_1^\prime}
~&~ \Omega^3(M,\h^{2})\\
\oplus ~&~\oplus ~&~\oplus ~&~\oplus \\
\Omega^0(M,\h^{2}) \ar[r]^-{\ell_1^\prime}~&~ \Omega^1(M,\h) \ar[r]^-{\ell_1^\prime}~&~ \Omega^3(M,\h) 
\ar[r]^-{0}~&~ 0
}
\end{gathered}
\right)\quad,
\end{flalign}
where the columns run from cohomological degree $0$ to $3$. This cochain complex
takes the same shape as the ones from Examples \ref{ex:complex1} and \ref{ex:complex2}, 
so we can apply the same interpretations to this theory. A remarkable feature of the present example
is that, in contrast to the two previous examples, the corresponding systems of differential equations
are well-posed. Before providing a precise statement and proof for well-posedness, let us start with a heuristic degree counting
argument. The top row describes a system consisting effectively of $3\times \dim(\h)\times 1-1\times \dim(\h)\times 2=\dim(\h)$
differential equations for the $1\times\dim(\h)\times 1 =\dim(\h)$ components of the scalar
field $\Phi \in \Omega^0(M,\h) $. Similarly, the bottom row describes
a system of $1\times\dim(\h)\times 1=\dim(\h)$ differential equations
for the $3\times\dim(\h)\times 1 - 1\times\dim(\h)\times 2 = \dim(\h)$ non-gauge redundant
components of the gauge field $\lambda\in \Omega^1(M,\h)$.
\end{ex}

To establish a precise statement and proof for the well-posedness 
of the dynamics described by Example \ref{ex:complex3}, it is necessary to
provide an explicit description of the differential $\ell^\prime_1$ 
of the cochain complex \eqref{eqn:complex3}. From the definition of this complex 
in \eqref{eqn:bdyM}, one observes that this requires explicit models for 
continuous strong deformation retracts of divisor-twisted Dolbeault complexes
to their cohomologies. Starting from Hodge theory, it is possible 
to develop such explicit strong deformation retracts in terms of explicit,
but rather sophisticated, integral formulas over $C=\CP$, 
which will be presented in detail in an upcoming publication \cite{BCSV}.
Applying these integral formulas to the cochain complex \eqref{eqn:complex3},
one finds an explicit description for the differential $\ell^\prime_1$.
To simplify our presentation, we consider only the case where $M=\bbR^3$
is the $3$-dimensional Cartesian space with coordinates denoted by $u^1,u^2,u^3$.
Relative to these coordinates, we describe $1$-forms $\sum_{i}\alpha_i \,\dd u^i$,
$2$-forms $\sum_{i<j}\beta_{ij} \,\dd u^i\wedge\dd u^j$ and $3$-forms
$\gamma \,\dd u^1 \wedge \dd u^2\wedge \dd u^3$ in terms of
their components. We then denote the components of the fields in the top and bottom row of \eqref{eqn:complex3} by
\begin{subequations}\label{eqn:fieldcomponents}
\begin{flalign}
\Phi \,\in\,\Omega^{0}(M,\h)\quad,\qquad \begin{pmatrix}\lambda^{+}_{12}\\ \lambda^{+}_{13}\\ \lambda^{+}_{23} \end{pmatrix}\,\in \,\Omega^2(M,\h)\quad,\qquad \begin{pmatrix}\epsilon^{+}\\ \chi^{+} \end{pmatrix} \,\in\, \Omega^3(M,\h^{2}) \quad,\\
\begin{pmatrix}\epsilon\\ \chi \end{pmatrix}\,\in\,\Omega^{0}(M,\h^{2})\quad,\qquad
\begin{pmatrix} \lambda_1\\ \lambda_2\\ \lambda_3 \end{pmatrix}\,\in\,\Omega^1(M,\h)\quad,\qquad
\Phi^{+}\,\in\,\Omega^3(M,\h) \quad,
\end{flalign}
\end{subequations}
where ${}^+$ is used to indicate anti-fields. From the integral formulas
which will appear in \cite{BCSV}, one can compute the four non-trivial components 
of the differential in \eqref{eqn:complex3}. Using the short-hand notation 
$\partial_i := \frac{\partial}{\partial u^i}$ for $i \in \{ 1,2,3\}$, one finds
\begin{subequations}\label{eqn:differentialcomponents}
\begin{flalign}
\ell^\prime_1\begin{pmatrix} \Phi \end{pmatrix}\,&=\,\begin{pmatrix}
(q_1 - q_2)\, \partial_1 \partial_2 \Phi\\[3pt] (q_1 - q_3) \,\partial_1 \partial_3 \Phi\\[3pt] (q_2 - q_3) \,\partial_2 \partial_3 \Phi
\end{pmatrix}\quad,\\
\ell^\prime_1 \begin{pmatrix} 
\lambda^{+}_{12}\\ \lambda^{+}_{13}\\ \lambda^{+}_{23}
\end{pmatrix}\,&=\,\begin{pmatrix} \partial_1 \lambda^{+}_{23} - 
\partial_2 \lambda^{+}_{13} + \partial_3 \lambda^{+}_{12}\\[3pt]
- q_1 \,\partial_1 \lambda^{+}_{23} + q_2\, \partial_2 \lambda^{+}_{13} - q_3 \,
\partial_3 \lambda^{+}_{12} \end{pmatrix}\quad,\\
\ell^\prime_1 \begin{pmatrix}\epsilon\\ \chi \end{pmatrix}\,&=\, 
\begin{pmatrix} - \tfrac 12 \,\partial_1 (\chi - q_1 \,\epsilon) \\[3pt] - \tfrac 12 \,\partial_2 (\chi - q_2\, \epsilon)\\[3pt] 
- \tfrac 12 \,\partial_3 (\chi - q_3 \,\epsilon) \end{pmatrix}\quad,\\
\ell^\prime_1 \begin{pmatrix} 
\lambda_1\\ \lambda_2\\ \lambda_3
\end{pmatrix}\,&=\, (q_1 - q_2) \, \partial_1 \partial_2 \lambda_3 - (q_1 - q_3)\, \partial_1 \partial_3 \lambda_2 
+ (q_2 - q_3) \, \partial_2 \partial_3 \lambda_1 \quad,
\end{flalign}
\end{subequations}
where $q_1,q_2,q_3 \in \bbC$ are the three zeros of the meromorphic $1$-form $\omega$
from Example \ref{ex:example3} (for $d=2$), written in terms of a choice of global holomorphic coordinate on $\CP\setminus\{\infty\}$.
\sk

We will now show that, in the case where the three zeros are real $q_1,q_2,q_3\in\bbR$ and 
distinct, i.e.\ $q_i\neq q_j$ for all $i\neq j$,
the complex of differential operators \eqref{eqn:complex3}, with differential
given explicitly by \eqref{eqn:differentialcomponents}, admits a Green's witness in
the sense of \cite{BMS}. By the general results of the latter paper, this implies in particular
the existence of a homological variant of Green's operators (called Green's homotopies) controlling
the dynamics and solutions of this theory. A Green's witness is defined as
a linear differential operator $W : \FFF(M)\to \FFF(M)$ of degree $-1$ with the property that
its graded commutator $W\,\ell_{1}^\prime + \ell_1^{\prime}\, W =: P$ with the differential
$\ell_1^\prime$ yields a Green-hyperbolic differential operator $P$, i.e.\ a differential
operator admitting Green's operators. In our example, a Green's witness is given by the 
following four non-trivial components
\begin{subequations}\label{eqn:Wexample}
\begin{flalign}
W\begin{pmatrix} \epsilon^+ \\ \chi^+ \end{pmatrix}\,&=\,\begin{pmatrix}
q_2 \,\partial_2 \epsilon^+ + \partial_2 \chi^+ + q_1 \,\partial_1 \epsilon^+ + \partial_1 \chi^+ \\[3pt]
- q_1 \,\partial_1 \epsilon^+ - \partial_1 \chi^+ + q_3\, \partial_3 \epsilon^+ + \partial_3 \chi^+ \\[3pt] 
- q_2 \,\partial_2 \epsilon^+ - \partial_2 \chi^+ - q_3 \,\partial_3 \epsilon^+ - \partial_3 \chi^+ \end{pmatrix}\quad,\\
W\begin{pmatrix} \lambda_{12}^+\\ \lambda_{13}^+\\ \lambda_{23}^+ \end{pmatrix}\,&=\, \lambda_{12}^+ + \lambda_{13}^+ + \lambda_{23}^+\quad,\\
W\begin{pmatrix} \Phi^+ \end{pmatrix}\,&=\,
\begin{pmatrix} \Phi^+\\ -\Phi^+\\ \Phi^+ \end{pmatrix}\quad,\\
W\begin{pmatrix} \lambda_{1}\\ \lambda_{2}\\ \lambda_{3} \end{pmatrix}\,&=\,\begin{pmatrix}
2\, \big(\partial_2 \lambda_1 + \partial_3 \lambda_1 
- \partial_1 \lambda_2 + \partial_3 \lambda_2 
- \partial_1 \lambda_3 - \partial_2 \lambda_3 \big) \\[3pt]
2\, \big(q_2\,\partial_2 \lambda_1 + q_3\,\partial_3 \lambda_1 
- q_1\,\partial_1 \lambda_2 + q_3\, \partial_3 \lambda_2 
- q_1\,\partial_1 \lambda_3 - q_2\,\partial_2 \lambda_3 \big)
\end{pmatrix}\quad,
\end{flalign}
\end{subequations}
which act horizontally along the rows of \eqref{eqn:complex3}.
By a direct calculation, one finds that the corresponding
differential operator $P:=W\,\ell_{1}^\prime + \ell_1^{\prime}\, W$ 
is given on all individual components by
\begin{subequations}\label{eqn:Poperator}
\begin{flalign}
P \,=\,\sum_{i,j=1}^{3} g^{ij}\,\partial_i\,\partial_j\quad,
\end{flalign} 
where $g^{ij}$ are the entries of the symmetric real $3\times 3$-matrix
\begin{flalign}
g \,=\, \frac{1}{2} \,\begin{pmatrix}
0 & q_1-q_2 & q_1-q_3 \\
q_1-q_2 & 0 & q_2-q_3\\
q_1-q_3 & q_2-q_3 & 0
\end{pmatrix}\quad.
\end{flalign}
\end{subequations}
From the fact that $\mathrm{Tr}(g)=0$ and $\det(g) = \frac{1}{4}\,(q_1-q_2)\,(q_1-q_3)\,(q_2-q_3) \neq 0$,
since by hypothesis $q_i\neq q_j$ for all $i\neq j$, it follows that all three eigenvalues of $g$ are non-zero
and that one eigenvalue must have the opposite sign of the other two. Assuming that $\det(g) < 0$,
which happens if the zeros of $\omega$ are ordered as $q_3 > q_2 > q_1$ or any cyclic permutation thereof,
then one eigenvalue is negative and the other two are positive. This implies that $g$ defines a Lorentzian metric on
spacetime $M=\bbR^3$ of signature $(-++)$. It is quite remarkable that the
spacetime geometry on $M=\bbR^3$ emerges from the singularity structure on $C=\CP$ in this way. By diagonalizing 
the matrix $g$ and rescaling the coordinates, one can transform from the null-coordinates
$u^1,u^2,u^3$ on $M=\bbR^3$ to standard Minkowski coordinates
\begin{flalign}
t \,=\, \frac{u^1 + u^2}{\sqrt{q_2 - q_1}} \quad, \quad
x \,=\, \frac{u^1 - u^2}{\sqrt{q_2 - q_1}} \quad, \quad
y \,=\, \frac{(q_3 - q_2)\, u^1 + (q_3 - q_1) \,u^2 - (q_2 - q_1)\, u^3}{\sqrt{(q_2 - q_1)\,(q_3 - q_1)\,(q_3 - q_2)}}\quad,
\end{flalign}
and thereby finds that 
the differential operator \eqref{eqn:Poperator} is nothing but the standard wave operator
$P = - \frac{\partial^2}{\partial t^2} + \frac{\partial^2}{\partial x^2} + \frac{\partial^2}{\partial y^2}$ 
on Minkowski spacetime $M=\bbR^3$.
Being a normally hyperbolic differential operator on Minkowski spacetime, 
$P$ admits unique retarded and advanced Green's operators
$G^\pm_P$, which determine, by using also our (square zero) Green's witness $W$ from \eqref{eqn:Wexample}, an explicit model for the
retarded and advanced Green's homotopies $\Lambda^\pm := W\,G^\pm_P = G^\pm_P\,W$ for the 
cochain complex \eqref{eqn:complex3}. As a consequence of the general results in \cite{BMS},
this implies that the dynamics, solution theory and even the commutator function entering
the quantization of our $3$-dimensional integrable field theory are all governed by an underlying wave equation.

%%%%%%%%%%%%%%%%%%%%%%%%%%%%%%%%%%%%%%%%%%%%%%%%
%%%%%%%%%%%%%%%%%%%%%%%%%%%%%%%%%%%%%%%%%%%%%%%%

\section*{Acknowledgments}
We would like to thank Severin Bunk for useful discussions.
The work of M.B.\ is supported in part by the MUR Excellence 
Department Project awarded to Dipartimento di Matematica, 
Universit{\`a} di Genova (CUP D33C23001110001) and it is fostered by 
the National Group of Mathematical Physics (GNFM-INdAM (IT)). 
R.A.C., A.S.\ and B.V.\ gratefully acknowledge the support of the Engineering and Physical
Sciences Research Council (UKRI1723).

%%%%%%%%%%%%%%%%%%%%%%%%%%%%%%%%%%%%%%%%%%%%%%%%
%%%%%%%%%%%%%%%%%%%%%%%%%%%%%%%%%%%%%%%%%%%%%%%%
%%%%%%%%%%%%%%%%%%%%%%%%%%%%%%%%%%%%%%%%%%%%%%%%
%%%%%%%%%%%%%%%%%%%%%%%%%%%%%%%%%%%%%%%%%%%%%%%%

%\section*{Data availability statement}
%All data generated or analyzed during this study are contained in this document. 
%
%
%\section*{Conflict of interest statement}
%On behalf of all authors, the corresponding author states that there is no conflict of interest.

%%%%%%%%%%%%%%%%%%%%%%%%

\end{document}